\documentclass[12pt]{article}
\usepackage[colorlinks, linkcolor=red, citecolor=blue]{hyperref}
\usepackage{amssymb}
\usepackage{amsmath}
\usepackage{amsthm}
\usepackage{cleveref}
\usepackage{enumerate}
\usepackage{cmll}
\usepackage{amsmath}

\allowdisplaybreaks[3]
\usepackage{bm}
\usepackage{natbib}
\usepackage{color}
\usepackage{float}
\usepackage{authblk}
\usepackage{graphicx}
\usepackage{algorithm}
\usepackage{algpseudocode}
\usepackage{epsfig}
\usepackage{dcolumn}
\usepackage{multirow}
\usepackage{cmll}
\usepackage{booktabs}
\usepackage{amsfonts}     
\usepackage{times}
\usepackage{bm}
\usepackage{natbib}
\usepackage{colonequals}
\usepackage{enumerate}
\usepackage{subfigure} 	

\graphicspath{{figures/}} 	

\usepackage{geometry}
\makeatletter

\newtheorem{theorem}{\bf{Theorem}}

\newtheorem{lemma}{\bf{Lemma}}

\newtheorem{example}{\bf{Example}}
\newtheorem{proposition}{\bf{Proposition}}
\newtheorem{condition}{\bf{Condition}}

\newcommand{\bbR}{\mathbb{R}}

\newcommand{\var}{{\rm var}}

\newcommand{\cM}{\mathcal{M}}

\newcommand{\cL}{\mathcal{L}}

\newcommand{\cS}{\mathcal{S}}

\newcommand{\cI}{\mathcal{I}}

\newcommand{\cK}{\mathcal{K}}

\graphicspath{{figures/}}

\def\T{{ \mathrm{\scriptscriptstyle T} }}

\title{A robust fusion-extraction procedure with summary statistics in the presence of biased sources}
\author{Ruoyu Wang, Qihua Wang\thanks{qhwang@amss.ac.cn}}
\affil{Academy of Mathematics and Systems Science, Chinese Academy of Sciences, Beijing 100190, China, and University of Chinese Academy of Sciences, Beijing 100049, China.}
\author{Wang Miao}
\affil{School of Mathematical Sciences, Peking University, Beijing 100871, China}
\date{}

\begin{document}
	\maketitle
	
	\begin{abstract}
		Information from multiple data sources is increasingly available. However, some data sources may produce biased estimates due to biased sampling, data corruption, or model misspecification. This calls for robust data combination methods with biased sources.  In this paper, a robust data fusion-extraction method is proposed. In contrast to existing methods, the proposed method can be applied to the important case where researchers have no knowledge of which data sources are unbiased. 	
		The proposed estimator is easy to compute and only employs summary statistics, and hence can be applied to many different fields, e.g., meta-analysis, Mendelian randomization, and distributed systems. The proposed estimator is consistent even if many data sources are biased and is asymptotically equivalent to the oracle estimator that only uses unbiased data. Asymptotic normality of the proposed estimator is also established. In contrast to the existing meta-analysis methods, the theoretical properties are guaranteed even if the number of data sources and the dimension of the parameter diverges as the sample size increases. Furthermore, the proposed method provides a consistent selection for unbiased data sources with probability approaching one. Simulation studies demonstrate the efficiency and robustness of the proposed method empirically. The proposed method is applied to a meta-analysis data set to evaluate the surgical treatment for moderate periodontal disease and to a Mendelian randomization data set to study the risk factors of head and neck cancer.
	\end{abstract}
	
	\noindent%
	\vfill

\section{Introduction}
In the big data era, it is common to have different data sources addressing a specific scientific problem of interest. An important question is how to combine these data to draw a final conclusion. In practice, due to privacy concerns or data transmission restrictions, individual-level data from different sources are usually not all available to the researcher. For some data sources, researchers only have access to certain summary statistics. To combine data information efficiently in this scenario, plenty of methods have been developed in the literature of meta-analysis, including the confidence distribution methods \citep{singh2005combining,xie2011confidence, liu2015multivariate}, the generalized method of moments (GMM), empirical likelihood-based methods \citep{sheng2020censored, qin2015using, chatterjee2016constrained,kundu2019generalized,zhang2019mendelian,zhang2020generalized}, and calibration methods \citep{lin2014adjustment,yang2020combining}. In many cases, estimators provided by meta-analysis methods are as efficient as the pooled estimators that use all the individual-level data \citep{olkin1998comparison, mathew1999equivalence, lin2010relative, xie2011confidence, liu2015multivariate}. Moreover, meta-analysis methods have been applied to large-scale data sets to reduce the computation and communication complexity \citep{jordan2013statistics,fan2014challenges,wang2016statistical}, even though all individual-level data are available in this scenario.

Ideally, all data sources can provide valid summary statistics (or consistent local estimates) based on data from them, respectively. Unfortunately, the summary statistics (or local estimates)  from some data sources may be invalid (or inconsistent) due to biased sampling, data corruption, model misspecification or other problems. Typical examples include the Simpson's paradox in meta-analysis \citep{hanley2000simpson}, the invalid instrument problem in Mendelian randomization \citep{qi2019mendelian,burgess2020robust} and Byzantine failure problem in distributed estimation \citep{lamport1982byzantine,yin2018byzantine,tu2021variance}.

\begin{example}[Simpson's paradox in meta-analysis] \label{ex: meta simpson}
	The dataset reported by \cite{hanley2000simpson} consists of five case-control studies that examine the role of high voltage power lines in the etiology of leukemia in children. \cite{hanley2000simpson} point out that different data fusion methods provide opposite conclusions. The reason is that three studies are conducted among the entire population, while two other studies undertake their investigation among the subpopulations living close to the power lines and thus suffer from biased sampling. Thus, two biased studies lead the final meta-analysis estimator to be biased. In this illustrative example, one knows which studies are biased and thus can just remove these studies. However, in practice, we seldom have such knowledge.
\end{example}

\begin{example}[Mendelian randomization with invalid instruments] \label{ex: invalid iv}
	In Mendelian randomization, single nucleotide polymorphisms (SNPs) are used as instrumental variables to evaluate the causal effect of a risk factor on the outcome. A SNP is a valid instrumental variable if (i) it is associated with the risk factor of interest; (ii) there is no confounder of the SNP-outcome association and (iii) it does not affect the outcome directly. See the Mendelian randomization dictionary \citep{lawlor2019mendelian} for more details. Suppose we have access to summary data representing the estimated effect of the $k$th SNP
	on the risk factor ($\tilde{\beta}_{k}$), and on the outcome ($\tilde{\gamma}_{k}$) for $k=1,\dots,K$. If the $k$th SNP  is a valid instrument, then $\tilde{\gamma}_{k}/\tilde{\beta}_{k}$ is a consistent estimator of the true causal effect. As data on several SNPs are available, we can use meta-analysis methods to produce a final estimator. However, in practice, a SNP may be an invalid instrument due to pleiotropy, linkage disequilibrium, and population stratification. In this case, $\tilde{\gamma}_{k}/\tilde{\beta}_{k}$ is no longer consistent and traditional meta-analysis methods can lead to biased estimates. In practice, however, it is hard to know which instrument is valid.	
\end{example}

\begin{example}[Byzantine failures in distributed systems]\label{ex: byzantine failure}
	To reduce the computational burden with large-scale data, the calculation of estimators is often conducted on distributed systems. Summary statistics are calculated on each local machine and transmitted to the central machine. The central machine produces the final estimator by combining these summary statistics. In practice, some local machines may produce wrong results and hence are biased sources due to hardware or software breakdowns, data crashes, or communication failures, which is called the ``Byzantine failures". Usually,  one does not know which machines have Byzantine failures.	Byzantine failures may deteriorate the performance of many divide-and-conquer methods.
\end{example}

As shown in the above three examples, many issues can result in biased estimation from a particular data source. In this paper, we call a data source biased if it produces an estimate that does not converge to the parameter of interest. Most of the aforementioned meta-analysis methods are not robust against the presence of biased sources, with the exception of \cite{singh2005combining,shen2020fusion} and \cite{zhai2022data}. 	Nevertheless, the method proposed in \cite{singh2005combining} is limited to the one-dimensional parameter case. Moreover, to apply the method of \cite{singh2005combining} and \cite{shen2020fusion}, there must be at least one known unbiased data source for reference. While the paper was submitted for review, \cite{zhai2022data} proposed a data fusion method based on the empirical likelihood, which can deal with summary statistics from biased data sources. Their method relies on the parametric conditional density model and requires individual-level data from a data source that is known to be unbiased. However, such an unbiased data source is often unavailable in practice. The main challenge for such a problem is that we do not know which data sources are biased, and hence one cannot remove the biased data sources from these data sources directly.	Clearly,  the use of biased data sources is adverse to defining a consistent estimator for the parameter of interest. This paper proposes a fusion-extraction procedure to define a consistent estimator and an asymptotically normal estimator, respectively, by combining all the summary statistics from	all the data sources in the presence of biased sources. 
In contrast to existing methods, the proposed method is applicable without any knowledge on which sources are unbiased. 

The proposed fusion-extraction procedure uses only summary statistics from different data sources and consists of two stages. In the first stage, we fuse the summary statistics from different data sources and obtain an initial estimator by minimizing the weighted Euclid distance from the estimators of all data sources. In the second stage, with the assistance of the initial estimator and the penalization method, we extract information from the unbiased sources and obtain the final estimator via a convex optimization problem. Both optimization problems in the two stages can be solved efficiently. Biased data sources do not affect the consistency of our method, as long as the proportion of unbiased sources among all sources is not too small. Moreover, our method can be implemented without knowledge on which data sources are unbiased. This makes our method more practical. The theoretical properties of our estimator are investigated under some mild conditions. We first show the consistency of the initial estimator produced by the first stage. Then we show that with this initial estimator, the final estimator produced by the second stage optimization is close in terms of Euclid norm to the oracle estimator that uses only unbiased data sources. Furthermore, it is shown that the extraction procedure in the second stage can consistently select unbiased sources. Based on these ``oracle properties", the asymptotic normality of the proposed estimator is also established. The established theorems are general in the sense that the number of data sources, $K$, and the dimension of parameter, $d$, can diverge as sample size increases. To our knowledge, no existing literature considers the meta-analysis problem in the presence of unknown biased sources when both $K$ and $d$ diverge. 

Our method is robust to biased data sources and computationally simple, and hence can be applied to many different fields, e.g. meta-analysis, Mendelian randomization, and distributed system. Besides its generality, it has some attractive properties across different fields. In contrast to existing works in meta-analysis with heterogeneous data sources \citep{singh2005combining,shen2020fusion,zhai2022data}, the proposed method does not require knowledge on which data sources are unbiased. In the field of Mendelian randomization with invalid instruments \citep{han2008detecting}, our method does not require that at least half of the instruments are valid while allowing for multiple or diverging number of treatments. Furthermore, our method can also be applied to the distribution system with Byzantine failures \citep{lamport1982byzantine,yin2018byzantine}. In contrast to the existing work \citep{tu2021variance}, the asymptotic normality of our method is guaranteed without requiring the proportion of biased sources among all sources to converge to zero. Moreover, our estimator has a faster convergence rate compared to that of \cite{tu2021variance} if the proportion of biased sources does not converge to zero.

Simulations under different scenarios demonstrate the efficiency and robustness of the proposed method. The proposed method is applied to a meta-analysis data set \citep{berkey1998meta} to evaluate the surgical treatment for moderate periodontal disease, and a Mendelian randomization data set \citep{gormley2020multivariable} to study the risk factors of head and neck cancer. The real data analysis results show the robustness of the proposed method empirically. 

The rest of this paper is organized as follows. In Section \ref{sec: estimation}, we suggest a two-stage method to provide an estimator for the parameter of interest in the presence of biased sources. In Section \ref{subsec: c&o}, we investigate the theoretical properties of the proposed estimator under certain general conditions on convergence rate. Under further assumptions, we establish the asymptotic normality of the proposed estimator in Section \ref{subsec: AN}. Simulation studies were conducted to evaluate the finite sample performance of our method in Section \ref{sec: simulation}, followed by two real data examples in Section \ref{sec: real data analysis}. Further simulation studies and all proofs are deferred to the supplementary material due to limited space.

\section{Estimation in the presence of biased sources}\label{sec: estimation}
\subsection{Identification}\label{subsec: identification}
Suppose $\theta_0$ is a $d$ dimensional parameter of interest and $K$ data sources can be used to estimate this parameter.  Each data source provides an estimator for the parameter of interest. Estimators from some data sources may be inconsistent for $\theta_{0}.$
The data sources that provide inconsistent estimators are called biased data sources but we do not know which data sources are biased. For $k=1,\dots,K$, let
$\tilde{\theta}_k$ be the estimate from the $k$th source, $n_k$ the sample size of the $k$th source, $n=\sum_{k=1}^{K}n_k$ and $\tilde{\pi}_{k} = n_{k}/n$. 
Estimates from different sources may be constructed using different methods and let $\theta^{*}_{k}$ be their probability limits respectively, i.e., $\|\tilde{\theta}_{k} - \theta^{*}_{k}\| \to 0 $ in probability for $k=1,\dots,K$, where $\|\cdot\|$ is the Euclid norm. Then a data source is biased if $\theta^{*}_{k} \neq \theta_0$. We assume that some of the sources are unbiased in the sense that $\theta^{*}_{k} = \theta_0$ but we do not know which are unbiased. Let $\cK_0=\{k:\theta^{*}_{k}=\theta_0\}$ be the set of indices of unbiased sources and $b^{*}_{k} = \theta^{*}_{k} - \theta_{0}$ be the bias of source $k$. Throughout this paper, we assume $\|b_{k}^{*}\|$ is bounded away from zero for $k\in \cK_{0}^{c} = \{1,\dots,K\}\setminus \cK_{0}$. 
Since we do not have any knowledge about $\cK_{0}$, we do not know whether $\theta_{k}^{*}$ equals to the parameter of interest or not.
Fortunately, the following proposition shows that $\theta_{0}$ can be identified as a weighted geometric median of $\theta_{k}^{*}$ for $k=1,\dots,K$ if the  proportion of data from unbiased sources among all data is larger than a certain threshold. 
\begin{proposition}\label{prop: identification}
	If 
	\begin{equation}\label{eq: id cond}
		\sum_{k\in \cK_0}\tilde{\pi}_{k} > \Big\|\sum_{k\in \cK_0^{c}} \tilde{\pi}_{k}\frac{b^{*}_{k}}{\|b^{*}_{k}\|}\Big\|,
	\end{equation}	
	then
	\begin{equation}\label{eq: identification}
		\theta_0 = \arg \min_{\theta}\sum_{k=1}^K \tilde{\pi}_{k}\|\theta^{*}_{k} - \theta\|.
	\end{equation}
\end{proposition}

See the supplementary material for the proof of this proposition.
The proposition shows that $\theta_{0}$ can be uniquely determined by $\theta_{k}^{*}$ for $k = 1,\dots, K$ if \eqref{eq: id cond} holds.
Note that
\begin{equation}\label{eq: proportion version}
	\Big\|\sum_{k\in \cK_{0}^{c}} \tilde{\pi}_{k}\frac{b^{*}_{k}}{\|b^{*}_{k}\|}\Big\| \leq \sum_{k\in \cK_{0}^{c}} \tilde{\pi}_{k}.
\end{equation}
Thus a sufficient condition for \eqref{eq: id cond} is
\[\sum_{k\in \cK_{0}}\tilde{\pi}_{k} > \sum_{k\in \cK_0^{c}} \tilde{\pi}_{k}\]
or equivalently
\begin{equation}\label{eq: majority rule}
	\sum_{k\in \cK_{0}}\tilde{\pi}_{k} > \frac{1}{2}.
\end{equation}
Inequality \eqref{eq: majority rule} requires more than half of the data come from unbiased data sources, which is related to the \emph{majority rule} widely adopted
in the invalid instrument literature \citep{kang2016instrumental,bowden2016consistent,windmeijer2019use}.
The previous analysis implies that \eqref{eq: id cond} is true under \eqref{eq: majority rule}. 
Next, we illustrate that \eqref{eq: id cond} can still hold even though less than a half of the data come from unbiased sources with a toy example. Suppose $d = 3$, $K = 6$, $\theta_{0} = (0, 0, 0)^{\T}$, $\tilde{\pi}_{k} = 1 / 6$, $b^{*}_{1} = b^{*}_{2} = (0, 0, 0)^{\T}$, $b_{3}^{*} = (1, 0, 0)^{\T}$, $b_{4}^{*} = (-2, 0, 0)^{\T}$, $b_{5}^{*} = (0, 1, 0)^{\T}$ and $b_{6}^{*} = (0, 0, 1)^{\T}$. In this case, $\sum_{k\in \mathcal{K}_{0}}\tilde{\pi}_{k} = 1/3 < 1/2$, however, $\|\sum_{k\in \cK_0^{c}} \tilde{\pi}_{k}b^{*}_{k}/\|b^{*}_{k}\|\| = \surd{2}/6 < 1/3$ and hence \eqref{eq: id cond} is satisfied. In Section \ref{sec: simulation}, we provide a further example where only $20\%$ of the data come from unbiased sources and \eqref{eq: id cond} still holds. Theoretically, if the equality in \eqref{eq: proportion version} holds, \eqref{eq: id cond} is equivalent to \eqref{eq: majority rule}; otherwise \eqref{eq: id cond} is weaker than \eqref{eq: majority rule}. The equality in \eqref{eq: proportion version} holds only if all $b^{*}_{k}$'s lie occasionally on the same direction, which is rarely true because the biases are often irregular in practice. 

In Proposition \ref{prop: identification}, the quantity $\tilde{\pi}_{k}$ can be viewed as the weight attached to the $k$-th data source for $k=1,\dots,K$.
It is observed in meta-analysis that small studies tend to be of relatively low methodological quality and are more likely to be affected by publication and selection bias \citep{sterne2000publication}.
Thus we use the weights $\{\tilde{\pi}_{k}\}_{k=1}^{K}$ that attach small weights to data sources with small sample sizes in this paper. 

By replacing $\theta_{k}^{*}$ by $\hat{\theta}_{k}$ in equation \eqref{eq: identification}, we can construct an estimator for $\theta_{0}$.
Further, we use the defined estimator as an initial estimator to obtain a more efficient estimator.

	
\subsection{Estimation}\label{subsec: estimation}
According to Proposition \ref{prop: identification}, we propose the following estimator $\tilde{\theta}$ that minimizes the weighted distance from $\tilde{\theta}_{k}$,
\begin{equation}\label{eq: estimation tilde}
	\tilde{\theta}=\arg\min_\theta \sum_{k=1}^{K}\tilde{\pi}_{k}\|\tilde{\theta}_k - \theta\|.
\end{equation}
This optimization problem is convex and can be solved efficiently by routine algorithms.
We show the consistency of $\tilde{\theta}$ in Section 
\ref{sec: main theory}. However, according to the well-known trade-off between robustness and efficiency \citep{Hampel1986, LindsayRobustness1994}, the robust estimator $\tilde{\theta}$ may not be fully efficient. Also, it may have a large finite sample bias because $\tilde{\theta}$ uses summary statistics from biased sources. Our simulations confirm this. The large finite sample bias implies that $\tilde{\theta}$ may not be $n^{1/2}$-consistent or asymptotically normal. Here we give a concrete example. Suppose $d=1$, $K \to \infty$, $n_{k} = n/K$, $\tilde{\theta}_{k} \sim N(\theta_{k}^{*}, 1/n_{k})$ for $k=1,\dots,K$ and $\tilde{\theta}_{k}$'s are independent of each other. Assume $\theta_{0} = 0$, $\theta_{k}^{*} = \theta_{0} + b_{k}^{*}$, $ b_{k}^{*} = 0$ for $k = 1,\dots, \lfloor(1/2+\tau)K\rfloor$ and $b_{k}^{*} = 1$ for $k = \lfloor(1/2+\tau)K\rfloor + 1, \dots, K$, where $0<\tau<1/2$. In the supplementary material, we show 
\begin{equation}\label{eq: counter ex}
	P\left(\tilde{\theta} - \theta_{0} \geq \frac{K^{1/2}h_{*}}{n^{1/2}}\right) \to 1,
\end{equation}
where $h_{*} = \Phi^{-1}((3/8 + \tau/4)/(1/2 + \tau))$ and $\Phi$ is the cumulative distribution function of standard normal distribution.
Because $h_{*} > 0$ and $K \to \infty$, \eqref{eq: counter ex} implies $\tilde{\theta}$ is not $n^{1/2}$-consistent.

Besides the aforementioned issue, the covariance structure of $\tilde{\theta}_{k}$ for $k=1,2,..., K$ is not considered in the construction of the estimator $\tilde{\theta}$ and this may lead to a loss of efficiency. These facts motivate us to propose an estimator which is not only $n^{1/2}$-asymptotically normal but also more efficient by 
using penalization technique and incorporating covariance structures of
$\tilde{\theta}_k$ for $k=1,2,...,K$.

It is well known that the oracle inverse-variance weighted (IVW) estimator
\begin{equation}\label{eq: oracle IVW}
	\hat{\theta}_{\rm IVW} = \arg\min_{\theta}\sum_{k\in \cK_0}\frac{\tilde{\pi}_{k}}{2}(\tilde{\theta}_{k} - \theta)^\T\tilde{V}_{k}(\tilde{\theta}_{k} - \theta)
\end{equation}
is the most efficient meta-analysis estimator and asymptotically normal if $\cK_{0}$ is known and $n_{k}^{1/2}(\tilde{\theta}_{k} - \theta_{k}^{*}) \to N(0, \Sigma_{k})$ for each $k$ in $\cK_{0}$ \citep{lin2010relative,xie2011confidence,burgess2020robust}, where $\tilde{V}_{k}$ is a consistent estimator of $\Sigma_{k}^{-1}$ for $k\in \cK_{0}$. See \cite{shen2020fusion} for further discussion on the optimality of $\hat{\theta}_{\rm IVW}$. In general, $\tilde{V}_{k}$'s can be any positive definite matrices if an estimator for $\Sigma_{k}^{-1}$ is unavailable \citep{liu2015multivariate} and $\hat{\theta}_{\rm IVW}$ is still $\sqrt{n}$ consistent and asymptotically normal under certain regularity conditions.
However, $\hat{\theta}_{\rm IVW}$ is infeasible if $\cK_{0}$ is unknown. Next, we develop a penalized inverse-variance weighted estimation method, which obviates the need to know $\cK_{0}$ and is asymptotically equivalent to $\hat{\theta}_{\rm IVW}$ under mild conditions.

To obtain a feasible estimator, we first replace $\cK_{0}$ with $\{1,\dots,K\}$ and obtain the objective function
\begin{equation}\label{eq: IVW objective}
	\sum_{k = 1}^{K}\frac{\tilde{\pi}_{k}}{2}(\tilde{\theta}_{k} - \theta)^\T\tilde{V}_{k}(\tilde{\theta}_{k} - \theta).
\end{equation}
Simply minimizing \eqref{eq: IVW objective} with respect to $\theta$ may produce an inconsistent estimator because some of $\tilde{\theta}_{k}$'s may not converge to $\theta_{0}$. Noticing that $\|\tilde{\theta}_{k} - \theta_{0} - b_{k}^{*}\| \to 0$ in probability for $k=1,\dots,K$, we add bias parameters $b_{k}$'s to \eqref{eq: IVW objective} and get
\begin{equation}\label{eq: IVW objective-bias}
	\sum_{k = 1}^{K}\frac{\tilde{\pi}_{k}}{2}(\tilde{\theta}_{k} - \theta - b_{k})^\T\tilde{V}_{k}(\tilde{\theta}_{k} - \theta - b_{k}).
\end{equation}
One may expect to recover $\theta_{0}$ and $b_{1}^{*}, \dots, b_{K}^{*}$  by minimizing \eqref{eq: IVW objective-bias} with respect to $\theta, b_{1}, \dots, b_{K}$. Unfortunately, it is not the case because, for any given $\theta$, \eqref{eq: IVW objective-bias} is minimized as long as $b_{k}$ takes $\tilde{\theta}_{k} - \theta$ for $k = 1,\dots, K$. Hence $\theta, b_{1}, \dots, b_{K}$ that minimize \eqref{eq: IVW objective-bias} are not necessarily close to $\theta_{0}, b_{1}^{*}, \dots, b_{K}^{*}$. To resolve this problem, we leverage the fact that $b_{k}^{*} = 0$ for $k\in \cK_{0}$ and impose penalties on $b_{k}$ to force $b_{k}$ to be zero and leave $b_{k}$ for $k \in \cK_{0}^{c}$ unconstrained. Hence we want to impose a large penalty on $b_{k}$ for $k \in \cK_{0}$ and impose no or a small penalty on $b_{k}$ for $k \in \cK_{0}^{c}$. To this end, we make use of the consistent estimator $\tilde{\theta}$ and define the following estimator    
\begin{equation}\label{eq: opt problem}
	(\hat{\theta}^{\T}, \hat{b}_{1}^{\T}, \dots, \hat{b}_{K}^{\T})^{\T} \in \mathop{\arg\min}_{\theta, b_1, \dots, b_{K}}\sum_{k=1}^{K}\left\{\frac{\tilde{\pi}_{k}}{2}(\tilde{\theta}_{k} - \theta - b_{k})^\T\tilde{V}_{k}(\tilde{\theta}_{k} - \theta - b_{k}) + \lambda\tilde{w}_{k}\|b_{k}\|\right\},
\end{equation}
where $\tilde{w}_k = \|\tilde{b}_{k}\|^{-\alpha}$, $\tilde{V}_{k}$ is some weighting matrix and $\lambda$ is a tuning parameter with $\alpha > 0$ and $\tilde{b}_{k} = \tilde{\theta}_{k} - \tilde{\theta}$ being an initial estimator of $b_{k}^{*}$. For $k\in \cK_{0}$, $\tilde{w}_{k}$ tends to be large because $\tilde{b}_{k} \to 0$ in probability. Thus $b_{k}$ may be estimated as zero in \eqref{eq: opt problem} for $k \in \cK_{0}$. On the other hand, because $\|\tilde{b}_{k}\| \to \|b_{k}^{*}\| > 0$ for $k \in \cK_{0}^{c}$, a smaller penalty is imposed on $b_{k}$ for $k\in \cK_{0}^{c}$ compared to $k \in \cK_{0}$. The optimization problem in \eqref{eq: opt problem} produces a continuous solution and is computationally attractive due to its convexity \citep{zou2006adaptive}. We propose $\hat{\theta}$ as an estimator for $\theta_{0}$. The form of \eqref{eq: opt problem} is akin to the adaptive Lasso \citep{zou2006adaptive} and the group Lasso \citep{yuan2006model}. The optimization problem in \eqref{eq: opt problem} can be rewritten as an adaptive group lasso problem and solved efficiently by the R package \emph{ggLasso} (\texttt{https://cran.r-project.org/web/packages/gglasso/index.html}). It is noted that $\tilde{\theta}$ makes contribution to $\hat{\theta}$ through  $\tilde{w}_k$.
This may help to select the estimates from unbias sources and control the bias of $\hat{\theta}$. We show in Section \ref{sec: main theory} that this $\hat{\theta}$ performs as well as the oracle estimator  $\hat{\theta}_{\rm IVW}$. 

\subsection{Implementation in examples}
Equations \eqref{eq: estimation tilde} and \eqref{eq: opt problem} provide two general estimating procedures in the presence of biased data sources and can be applied to many specific problems. Only estimates $\tilde{\theta}_{k}$ from different data sources are required to conduct the proposed procedure. Specifically, in Example \ref{ex: meta simpson}, we can take $\tilde{\theta}_{k}$ to be the estimate from the $k$th study for $k=1,\dots,5$. In Example \ref{ex: invalid iv}, we let $\tilde{\theta}_{k} = \tilde{\gamma_{k}}/\tilde{\beta_{k}}$ and use the proposed procedure to deal with the invalid instrument problem. In Example \ref{ex: byzantine failure}, we use the output of each local machine as $\tilde{\theta}_{k}$'s and apply our method to mitigate the effects of Byzantine failures. We investigate the theoretical properties of the proposed fusion-extraction procedure in the next section since they are of wide application.
%

\section{Theoretical properties}\label{sec: main theory}
\subsection{Consistency and oracle properties}\label{subsec: c&o}
In this subsection, we provide asymptotic results for the estimators proposed in Section \ref{subsec: estimation}. In our theoretical development, both the dimension of parameter $d$ and the number of sources $K$ are allowed to diverge as $n \to \infty$. Let
\begin{align*}
	\delta = \sum_{k\in \cK_0}\tilde{\pi}_{k} - \Big\|\sum_{k\in \cK_0^{c}} \tilde{\pi}_{k}\frac{b^{*}_{k}}{\|b^{*}_{k}\|}\Big\|.
\end{align*}
Then \eqref{eq: id cond} is equivalent to $\delta > 0$ and we have the following theorem.
\begin{theorem}\label{thm: consistent}
	If $\delta > 0$ and $\delta^{-1}\max_{k}\|\tilde{\theta}_{k} - \theta^{*}_{k}\| \to 0$ in probability, then $\|\tilde{\theta} - \theta_{0}\| \to 0$ in probability.
\end{theorem}

Proof of this theorem is relegated to the supplementary material.
Theorem \ref{thm: consistent} establishes the consistency of $\tilde{\theta}$ under the condition that $\delta$ is not too small and that $\tilde{\theta}_{k}$ converges uniformly for $k=1,\dots,K$. If $K$ is fixed, the condition $\delta^{-1}\max_{k}\|\tilde{\theta}_{k} - \theta^{*}_{k}\| \to 0$ in probability can be satisfied as long as $\delta$ is positive and bounded away from zero. Having established the theoretical property of the initial estimator $\tilde{\theta}$, next we investigate the theoretical properties of $\hat{\theta}$ defined in \eqref{eq: opt problem}. As pointed out previously, $\hat{\theta}_{\rm IVW}$ is an oracle estimate that uses summary data from unbiased sources only by combining them in an efficient way \citep{lin2010relative, xie2011confidence}. 
It is of interest to investigate how far away the proposed estimator $\hat{\theta}$ departs from the oracle estimator $\hat{\theta}_{\rm IVW}$ is.
To establish the convergence rate of $\|\hat{\theta} - \hat{\theta}_{\rm IVW}\|$, the following conditions are required. For any symmetric matrix $A$, let  $\lambda_{\rm min}(A)$ and $\lambda_{\rm max}(A)$ be the minimum and maximum eigenvalue of $A$, respectively. We use $\|\cdot\|$ to denote the Euclid and spectral norm when applying to a vector and a matrix, respectively.
\begin{condition}\label{cond: existence of nonbias}
	$\sum_{k \in \cK_{0}}\tilde{\pi}_{k}$ is bounded away from zero.
\end{condition}

\begin{condition}\label{cond: matrix rate}
	There are some deterministic matrices $V_{k}^{*}$ ($k\in \cK_{0}$), such that $\max_{k\in \cK_{0}}\|\tilde{V}_{k} - V_{k}^{*}\| = o_{P}(1)$ where $\tilde{V}_{k}$ is the weighting matrix appearing in \eqref{eq: opt problem}. Moreover, the eigenvalues of $V_{k}^{*}$ are bounded away from zero and infinity for $k \in \cK_{0}$.
\end{condition}

\begin{condition}\label{cond: estimate rate}
	$K = O(n^{\nu_{1}})$, $\delta > 0$ and $\delta^{-1}\max_{k}\|\tilde{\theta}_{k} - \theta_{k}^{*}\| = O_{P}(n^{-\nu_{2}})$ for some $\nu_{1} \in [0,1)$ and $\nu_{2} \in (0,1)$. 
\end{condition}

Condition \ref{cond: existence of nonbias} assumes that the proportion of data from unbiased data sources is bounded away from zero, which is a reasonable requirement.
The weighting matrix $\tilde{V}_{k}$ may affect the performance of the resulting estimator. In many cases, the optimal choice of $\tilde{V}_{k}$ is shown to be the inverse of $\tilde{\theta}_{k}$'s asymptotic variance matrix \citep{lin2010relative,liu2015multivariate}. Condition \ref{cond: matrix rate} can be easily satisfied if the inverses of the estimated asymptotic variance matrices are used and $d$, $K$ are not too large \citep{vershynin2018high,wainwright2019high}. There are some difficulties in estimating the asymptotic variance matrix and its inverse when the dimension is high \citep{wainwright2019high}. 
In addition, sometimes the estimated asymptotic variance matrix is not available from the summary statistics \citep{liu2015multivariate}. However, Condition \ref{cond: matrix rate} just requires $\tilde{V}_{k}$ to converge to some nonsingular matrix. In these cases, we can simply take $\tilde{V}_{k}$ to be the identity matrix for $k=1,\dots,K$ and Condition \ref{cond: matrix rate} can always be satisfied. Condition \ref{cond: estimate rate} assumes the number of data sources $K$ is not too large. The convergence rate in Condition \ref{cond: estimate rate} can be satisfied by many commonly-used estimators, e.g., the maximum likelihood estimator and lasso-type estimators, under certain regularity conditions \citep{spokoiny2012parametric,battey2018distributed}. Then we are ready to state the theorem.

\begin{theorem}\label{thm: equivalence}
	Under Conditions \ref{cond: existence of nonbias}, \ref{cond: matrix rate} and \ref{cond: estimate rate}, if $\lambda \asymp 1/n$ and $\alpha > \max\{\nu_{1}\nu_{2}^{-1}, \nu_{2}^{-1} - 1\}$, we have
	\[\|\hat{\theta} - \hat{\theta}_{\rm IVW}\| = O_{P}\left(\frac{K}{n}\right).\]
\end{theorem}

Proof of this theorem is in the supplementary material.
Theorem \ref{thm: equivalence} establishes the convergence rate of $\|\hat{\theta} - \hat{\theta}_{\rm IVW}\|$, which indicates that  the proposed estimator is close to the oracle estimator. If $K = o(n^{1/2})$, then $\|\hat{\theta} - \hat{\theta}_{\rm IVW}\| = o_{P}(n^{-1/2})$ and our proposal is asymptotically equivalent to the oracle estimator up to an error term of order $o_{P}(n^{-1/2})$. The estimator proposed in \cite{shen2020fusion} possesses the similar asymptotic equivalence property. However, theoretical results in \cite{shen2020fusion} require $d$ and $K$ to be fixed, which is not required by Theorem \ref{thm: equivalence}. Moreover,  implementation of the estimator proposed in \cite{shen2020fusion} requires at least one known unbiased data source. In contrast, we do not need any information about $\cK_{0}$ to calculate $\hat{\theta}$. 

Theorem \ref{thm: equivalence} is generic in the sense that it only relies on some convergence rate conditions and does not impose restrictions on the form of $\tilde{\theta}_{k}$. In practice, $\tilde{\theta}_{k}$ may be calculated based on complex data, such as survey sampling or time-series data, via some complex procedure, such as deep learning or Lasso-type penalization procedure. In these cases, Theorem \ref{thm: equivalence} holds consistently as long as Conditions \ref{cond: matrix rate} and \ref{cond: estimate rate} are satisfied.  Moreover, Theorem \ref{thm: equivalence} does not require $\tilde{\theta}_{k}$'s to be independent of each other, which ensures validity of the theorem in meta-analysis with overlapping subjects \citep{lin2009meta} or one sample Mendelian randomization \citep{minelli2021use}. 

When solving \eqref{eq: opt problem}, we also get an estimator $\hat{b}_{k}$ of the bias. A question is whether $\{\hat{b}_{k}\}_{k=1}^{K}$ selects the unbiased sources consistently, that is, whether $\hat{\cK}_{0} = \cK_{0}$ with probability approaching one, where $\hat{\cK}_{0} = \{k:\hat{b}_{k} = 0\}$. To assure this selection consistency, a stronger version of Condition \ref{cond: matrix rate} is required.
\begin{condition}\label{cond: matrix rate 2}
	For some deterministic matrices $V_{k}^{*}$ ($k=1,\dots,K$), such that $\max_{k}\|\tilde{V}_{k} - V_{k}^{*}\| = o_{P}(1)$ where $\tilde{V}_{k}$ is the weighting matrix appears in \eqref{eq: opt problem}. Moreover, the eigenvalues of $V_{k}^{*}$ are bounded away from zero and infinity for $k = 1,\dots,K$.
\end{condition}
This condition requires that Condition \ref{cond: matrix rate} holds not only for $k\in \cK_{0}$ but also for $k \in \cK_{0}^{c}$, which is still a mild requirement. Then we are ready to establish the selection consistency.
\begin{theorem}\label{thm: selection consistency}
	Under Conditions \ref{cond: existence of nonbias}, \ref{cond: estimate rate} and \ref{cond: matrix rate 2}, if $\lambda \asymp 1/n$ and $\alpha > \max\{\nu_{1}\nu_{2}^{-1}, \nu_{2}^{-1} - 1\}$, we have
	\[P(\hat{\cK}_{0} = \cK_{0}) \to 1\]
	provided $\min_{k\in \cK_{0}^{c}}\tilde{\pi}_{k} > C_{\pi} /K$ and $K\log n /n \to 0$ where $C_{\pi}$ is some positive constant.
\end{theorem}
Proof of this theorem is relegated to the supplementary material.
\subsection{Asymptotic normality}\label{subsec: AN}
In this subsection, we establish the asymptotic normality of the proposed estimator $\hat{\theta}$. 
Under Conditions \ref{cond: existence of nonbias}, \ref{cond: matrix rate} and \ref{cond: estimate rate}, if $K = o(n^{1/2})$, then $\|\hat{\theta} - \hat{\theta}_{\rm IVW}\| = o_{P}(n^{-1/2})$. If $\hat{\theta}_{\rm IVW}$ is $n^{1/2}$-asymptotically normal, then $\hat{\theta}$ is $n^{1/2}$-asymptotically normal and has the same asymptotic variance as $\hat{\theta}_{\rm IVW}$. There exist some results on asymptotic normality of $\hat{\theta}_{\rm IVW}$ in the literature. However, these results either focus on the fixed dimension case \citep{lin2010relative,zhu2021least} or are only suited to some specific estimators under sparse linear or generalized liner model \citep{battey2018distributed}. Here, we establish the asymptotic normality of $\hat{\theta}_{\rm IVW}$, and hence of $\hat{\theta}$ in a general setting where $d$ and $K$ can diverge and $\tilde{\theta}_{k}$ ($k\in \cK_{0}$) can be any estimator that admits uniformly an asymptotically linear representation defined in the following. Suppose the original data from the $k$th source $Z_{1}^{(k)},\dots,Z_{n_{k}}^{(k)}$ are i.i.d. copies of $Z^{(k)}$ and the data from different data sources are independent from each other. Then we are ready to state the condition.
\begin{condition}[Uniformly asymptotically linear representation]\label{cond: uniform AL}
	For $k\in \cK_{0}$, there is some function $\Psi_{k}(\cdot)$ such that $E[\Psi_{k}(Z^{(k)})] = 0$ and
	\begin{equation}\label{eq: asymptotically linear representation}
		\tilde{\theta}_{k} - \theta_{0} = \frac{1}{n_{k}}\sum_{i=1}^{n_{k}} \Psi_{k}(Z_{i}^{(k)}) + R_{k},
	\end{equation}
	where $R_{k}$ satisfies $\max_{k}\|R_{k}\| = o_{P}(n^{-1/2})$.
\end{condition}

Some examples satisfying Condition 5 will be discussed later. With the assistance of the uniformly asymptotically linear representation condition (Condition \ref{cond: uniform AL}), we can establish the asymptotic normality of $\hat{\theta}_{\rm IVW}$ and hence of $\hat{\theta}$.
\begin{theorem}\label{thm: AN}
	Suppose Conditions \ref{cond: existence of nonbias}, \ref{cond: estimate rate} and \ref{cond: uniform AL} hold. If (i) $\nu_{1} < 1/2$; (ii) there are some deterministic matrices $V_{k}^{*}$, $k\in \cK_{0}$, such that $\max_{k\in \cK_{0}}\|\tilde{V}_{k} - V_{k}^{*}\| = o_{P}(n^{- 1 / 2 + \nu_{2}})$; (iii) for $k\in \cK_{0}$, the eigenvalues of $V_{k}^{*}$ and  $\var\left[\Psi_{k}(Z^{(k)})\right]$ is bounded away from zero and infinity; (iv) for $k\in \cK_{0}$, $u\in \bbR^{d}$, $\|u\| = 1$ and some $\tau > 0$, $E[|u^{\T}\Psi_{k}(Z^{(k)})|^{1 + \tau}]$ are bounded; (v) $\lambda \asymp 1/n$ and $\alpha > \max\{\nu_{1}\nu_{2}^{-1}, \nu_{2}^{-1} - 1\}$, then for any fixed $q$ and $q\times d$ matrix $W_{n}$ such that the eigenvalues of $W_{n}W_{n}^{\T}$ are bounded away from zero and infinity, we have
	\[n^{1/2}\cI_{n}^{-1/2}W_{n}(\hat{\theta} - \theta_{0})\to N(0, I_{q})\]
	in distribution, where $I_{q}$ is the identity matrix of order $q$, \[\cI_{n} = \sum_{k \in \cK_{0}}\tilde{\pi}_{k}H_{n,k}\var\left[ \Psi_{k}(Z^{(k)})\right]H_{n,k}^{\T}\]
	with $H_{n,k} = W_{n}V_{0}^{*-1}V_{k}^{*}$ and $V_{0}^{*} = \sum_{k\in \cK_{0}} \tilde{\pi}_{k}V_{k}^{*}$.
\end{theorem}

Proof of this theorem is in the supplementary material.
Many estimators have the asymptotically linear representation \eqref{eq: asymptotically linear representation} with $R_{k} = o_{P}(n_{k}^{-1/2})$, see for instance \cite{bickel1993efficient,spokoiny2013bernstein,zhou2018new} and \cite{chen2020robust}. For these estimators, Condition \ref{cond: uniform AL} is satisfied if remainder terms $R_{k}$'s are uniformly small for $k\in \cK_{0}$ in the sense that $\max_{k\in \cK_{0}}\|R_{k}\| = o_{P}(n^{-1/2})$. If $K$ is fixed, then $\max_{k\in \cK_{0}}\|R_{k}\| = o_{P}(n^{-1/2})$ as long as $\tilde{\pi}_{k}$'s are bounded away from zero. For the case where $K\to \infty$, in the supplementary material, we show that Condition \ref{cond: uniform AL} holds under some regularity conditions if $\tilde{\theta}_{k}$'s are M-estimators, i.e.
\[\tilde{\theta}_{k} = \mathop{\arg\min}_{\theta}\frac{1}{n_{k}}\sum_{i=1}^{n_{k}}L_{k}(Z_{i}^{(k)}, \theta),\]
for $k=1,\dots,K$, where $L_{k}(\cdot,\cdot)$ is some loss function that may differ from source to source. The result on M-estimator covers many commonly used estimators, e.g., the least squares estimator and maximum likelihood estimator. 

Theorem \ref{thm: AN} establishes the asymptotic normality of $\hat{\theta}$. Unlike the existing work that can deal with biased sources in the meta-analysis literature \citep{singh2005combining,shen2020fusion,zhai2022data}, both the proposed estimator $\hat{\theta}$ and $\tilde{\theta}$ can be obtained without any knowledge on $\cK_{0}$.  Compared to existing estimators in the literature of Mendelian randomization that focus on a one-dimensional parameter \citep{kang2016instrumental,bowden2016consistent,windmeijer2019use,hartwig2017robust,guo2018confidence, ye2021debiased}, the proposed $\hat{\theta}$ is applicable to the case where a multidimensional parameter is of interest. Moreover, the corresponding theoretical results
are more general in the sense that they allow for the divergence of both $d$ and $K$ as the sample size increases. Thus besides univariable Mendelian randomization, our method can also be applied to multivariable Mendelian randomization \citep{burgess2015multivariable,rees2017extending,sanderson2019examination} in the presence of invalid instruments. Recently, \cite{tu2021variance} developed a method that can deal with biased sources and also allows $d$ and $K$ to diverge. In contrast to their work, the estimator obtained by our method achieves the $n^{1/2}$-asymptotic normality without requiring the proportion of biased sources among all sources to converge to zero. According to the discussion after Proposition \ref{prop: identification}, the asymptotic normality of $\hat{\theta}$ is guaranteed even if more than half of the data come from biased sources. Therefore, our method is quite robust against biased sources and this is confirmed by our simulation results in the next section.

\section{Simulation}\label{sec: simulation}
In this section, we conducted three simulation studies to evaluate the empirical performance of the proposed methods. We consider different combinations of $d$ and $K$, with $d = 3, 18$ and $K =10, 30$. 

\subsection{Least squares regression}\label{subsec: sim-LS}
First, we consider the case where $\tilde{\theta}_{k}$'s are obtained via least squares. Let $1_{s}$ be the $s$ dimensional vector consisting of $1$'s and $\otimes$ be the Kronecker product. In this simulation,
the data from the $k$th source are generated from the following data generation process:\\
$X_{k} \sim N_{d}(0, 3I_{d})$,  $Y_{k}\mid X_{k} \sim N_{d}(X_{k}^{\T}(\theta_{0} + b_{k}^{*}), 1)$,
where $I_{d}$ is the identity matrix of order $d$, $\theta_{0} = 1_{d/3}\otimes(2, 1, -1)^{\T}$,  $(b_1^{*}, \dots, b_{K}^{*}) = 1_{K/10}^{\T}\otimes 1_{d/3}\otimes B$ and 
\begin{equation}\label{eq: bias matrix}
B = \left( \begin{array}{rrrrrrrrrr}
0 &0 &5 &-1&1 &1 &-2&-2&5 &-1\\
0 &0 &0 &0 &0 &-1&0 &2 &5 &-1\\
0 &0 &0 &0 &-1&1 &2 &-2&5 &1
\end{array}\right).
\end{equation}
From each data source, an i.i.d sample of size  $n_{*}$ is generated. We consider three different values of $n_{*}$, namely, $n_{*} = 100, 200$ or $500$.
In this simulation setting, only $20\%$ of the data come from unbiased sources. However, the biases do not lie in the same direction and it can be verified that $\delta > 0$ by straightforward calculations. Let $\tilde{\theta}_{k}$ be the least squares estimator from the $k$th data source. In the simulation, we simply take $\tilde{V}_{k}$ to be identity matrix to ensure that all the conditions on $\tilde{V}_{k}$ in this paper are satisfied.
We compute the naive estimator $\sum_{k = 1}^{K}\tilde{\theta}_{k} / K$, the oracle estimator $\hat{\theta}_{\rm IVW}$, the \emph{iFusion} estimator proposed by \cite{shen2020fusion}, the initial estimator $\tilde{\theta}$ and the proposed estimator $\hat{\theta}$. Note that the \emph{iFusion} estimator is infeasible unless at least one data source is known to be unbiased. In this section, we always assume that the first data source is known to be unbiased when computing the \emph{iFusion} estimator. This information is not required by $\tilde{\theta}$ and $\hat{\theta}$. The following table presents the norm of the bias vector (NB) and summation of the component-wise standard error (SSE) of these estimators calculated from 200 simulated data sets for all the four combinations of $d$ and $K$ with $d=3,18$ and $K = 10, 30$.

\centerline{[Insert Table~\ref{table: sim1} about here.]}

The naive estimator has a large bias which renders its small standard error meaningless. The bias of all the other estimators decreases as $n_{*}$ increases. The \emph{iFusion} estimator performs similarly to the oracle estimator when $d$ and $K$ are small. However, if $d = 18$ and $K = 30$, it has a much larger standard error compared to the oracle estimator and $\hat{\theta}$. The initial estimator $\tilde{\theta}$ performs well in terms of standard error. Nevertheless, it has a far larger bias compared to the oracle estimator and $\hat{\theta}$, especially when $d$ and $K$ are large. The reason may be that it is not $\sqrt{n}$-consistent. The performance of $\hat{\theta}$ is similar to the oracle estimator. This confirms the asymptotic equivalence between $\hat{\theta}$ and $\hat{\theta}_{\rm IVW}$ established in Section \ref{subsec: c&o}. 

Next, we evaluate the performance of our methods when all the data sources are unbiased. We set $b_{k}^{*}$ to be a zero vector for $k = 1,\dots, K$ while keeping other parameters unchanged. In this scenario, the naive estimator reduces to the oracle estimator. NB and SSE of the estimators calculated from 200 simulated data sets for all the combinations of $d$ and $K$ are summarized in the following table.

\centerline{[Insert Table~\ref{table: sim2} about here.]}

Table \ref{table: sim2} shows that the \emph{iFusion} estimator has a slightly larger bias and a much larger standard error compared to other estimators when $d = 18$. All other estimators have similar performance when there are no biased sources. This implies that there is little loss of efficiency to apply our methods when all the sources are unbiased.

\subsection{Logistic regression}
In this subsection, we conducted a simulation study under the scenario where the responses are binary and $\tilde{\theta}_{k}$'s are obtained via logistic regression. 
All simulation settings are the same as in Section \ref{subsec: sim-LS} except for that
$Y_{k}\mid X_{k} \sim {\rm Bernoulli}(t(X_{k}^{\T}(\theta_{0} + b_{k}^{*})))$ and $\tilde{\theta}_{k}$ is the maximum likelihood estimator of logistic regression model from the $k$th data source,
where $t(x) = \exp(x)/(1+\exp(x))$ is the logistic function, $\theta_{0} = 0.1 \times 1_{d/3}\otimes(2, 1, -1)^{\T}$,  $(b_1^{*}, \dots, b_{K}^{*}) = 0.5\times 1_{K/10}^{\T}\otimes 1_{d/3}\otimes B$ and $B$ is defined in \eqref{eq: bias matrix}. We add a small ridge penalty when solving $\tilde{\theta}_{k}$ to avoid the problem that the maximum likelihood estimator may not be uniquely determined in finite sample \citep{silvapulle1981existence}. 
NB and SSE of these estimators calculated from 200 simulated data sets are summarized in the following table.

\centerline{[Insert Table~\ref{table: sim3} about here.]}

The naive estimator has a large bias and the bias does not decrease as $n_{*}$ increases. The bias of all the other estimators decreases as $n_{*}$ increases. The \emph{iFusion} estimator has a bias similar to that of the oracle estimator. Its standard error is much larger than the oracle estimator and $\hat{\theta}$ especially when $d$, $K$ are large and $n_{*}$ is small.  The initial estimator $\tilde{\theta}$ has a much larger bias compared to the oracle estimator and the proposed estimator $\hat{\theta}$, which are consistent with the simulation results under least squares regression. The performance of $\hat{\theta}$ is similar to the oracle estimator. 

Next, we evaluate the performance of our methods when all the data sources are unbiased. We set $b_{k}^{*}$ to be the zero vector for $k = 1,\dots, K$ while keeping other parameters unchanged. The naive estimator reduces to the oracle estimator in this scenario. NB and SSE of the estimators calculated from 200 simulated data sets are summarized in the following table.

\centerline{[Insert Table~\ref{table: sim4} about here.]}

All the estimators have similar performance in Table \ref{table: sim4} except for that \emph{iFusion} estimator has a much larger standard error compared to other estimators, and there is little loss of efficiency to apply our methods when all the sources are unbiased.

\subsection{Mendelian randomization with invalid instruments}
We consider Mendelian randomization with invalid instruments in this subsection. To closely mimic what we will encounter in practice, we generate data based on a real-world data set, the BMI-SBP data set in the R package \emph{mr.raps} (version 0.4) of \cite{zhao2020statistical}. The data set contains estimates of the effects of $160$ different SNPs on Body Mass Index (BMI) $\{\bar{\beta}_{k}\}_{k=1}^{160}$ and the corresponding standard error $\{\bar{\sigma}_{1,k}\}_{k=1}^{160}$ from a study by the Genetic Investigation of ANthropometric Traits consortium \citep{locke2015genetic} (sample size: 152893), and estimates of the effects on Systolic Blood Pressure (SPB) $\{\bar{\gamma}_{k}\}_{k=1}^{160}$ and the corresponding standard error $\{\bar{\sigma}_{2,k}\}_{k=1}^{160}$ from  the UK BioBank (sample size: 317754). The goal is to estimate the causal effect of BMI on SPB. 

In this simulation, we generate data via the following process:
\[\tilde{\beta}_{k} \sim N(\bar{\beta}_{k}, \bar{\sigma}_{1,k}^{2}), \quad k= 1,\dots, 160,\] 
\[\tilde{\gamma}_{k} \sim N(\bar{\beta}_{k}\theta_{0} + 0.15 + 3\bar{\beta}_{k}, \bar{\sigma}_{2, k}^{2}),\quad k=1,\dots, 100\]
and 
\[\tilde{\gamma}_{k} \sim N(\bar{\beta}_{k}\theta_{0}, \bar{\sigma}_{2, k}^{2}), \quad k=101,\dots, 160\]
where $\theta_{0} = 1$. Then $\tilde{\theta}_{k}$ is given by $\tilde{\gamma}_{k}/\tilde{\beta}_{k}$. Under this data generation process, $100$ out of $160$ instruments are invalid instruments. We apply the proposed methods to estimate $\theta_{0}$ based on $\tilde{\theta}_{k}$'s.
For comparison,  we apply five standard methods in Mendelian randomization, namely the MR-Egger regression \citep{bowden2015mendelian}, the weighted median method \citep{bowden2016consistent}, the IVW method, the weighted mode method \citep{hartwig2017robust} and the robust adjusted
profile score method \citep[RAPS,][]{zhao2020statistical}. Results of these five methods are calculated by the R package \emph{TwoSampleMR} (\texttt{https://github.com/MRCIEU/ TwoSampleMR}). Bias and standard error (SE) of the estimators based on $200$ simulations are summarized in the following table.

\centerline{[Insert Table~\ref{table: sim5} about here.]}

Table \ref{table: sim5} shows that $\hat{\theta}$ has the smallest bias among all the estimators. The standard error of the proposed $\hat{\theta}$ is smaller than other estimators except for the weighted median. However, the weighted median estimator has a much larger bias compared to $\hat{\theta}$. 

\section{Real Data Analysis}\label{sec: real data analysis}
\subsection{Effects of surgical procedure for the treatment of moderate periodontal disease}
In this subsection, we apply our methods to the data set provided in \cite{berkey1998meta}. Data used in this subsection are available from the R package \emph{mvmeta} (\texttt{https://cran.r-project.org/web/ packages/mvmeta/index.html}). The data set contains results of five randomized controlled trials comparing the effect of surgical and non-surgical treatments for moderate periodontal disease. In all these studies, different segments of each patients' mouth were randomly allocated to different treatment procedures. The two outcomes, probing depth (PD) and attachment level (AL), were
assessed from each patient. The goal of treatment is to decrease PD and to increase
AL around the teeth. The data set provides the estimated benefit of surgical treatment over non-surgical treatment in PD and AL (positive values mean that
surgery results in a better outcome). The sample size of each study and estimated covariance matrix of the estimators are also available. The inverse-variance weighted method using all data sources produces an estimator $(0.307, -0.394)$ for the effect on $(\text{PD, AL})$. By applying our methods, we obtain $\tilde{\theta} = (0.260, -0.310)$ and $\hat{\theta} = (0.282, -0.303)$. Next, we assess the robustness of our methods against the bias of the published results. To do this, we add a perturbation $t\times(1, -3)$ to the first published result in the data set. After perturbing, the first published result becomes a biased estimator for the parameter of interest. We plot the resulting estimates with different values of $t$ in the following figure.

\centerline{[Insert Fig~\ref{fig: stable} about here.]}

Figure \ref{fig: stable} shows that $\tilde{\theta}$ and $\hat{\theta}$ provide quite stable estimates under different values of $t$ compared to the IVW estimator. Based on the reported estimated covariance matrices and the asymptotic normality of $\hat{\theta}$, we conduct two hypothesis testings to test whether these effects are significant. The $p$-values of the two-sided tests for effect on PD and AL are $1.660\times 10^{-7}$ and $5.278\times 10^{-15}$, respectively. This suggests that both effects are significant at $0.05$ significance level. The $95\%$ confidence intervals for effects of the surgical treatment on PD and AL based on $\hat{\theta}$ are $[0.176, 0.387]$ and $[-0.379,-0.227]$, respectively.  In summary, our result suggests that the surgical treatment has a positive effect on PD and a negative effect on AL and our result is robust against potential bias in the results of the five published trials.

\subsection{Effects of smoking and alcohol in head and neck cancer}
Head and neck cancer is the sixth most common cancer in the world. Established risk factor of this cancer includes smoking and alcohol. However, researchers only have a limited understanding of the causal effect of these risk factors due to the unmeasured confounding \citep{gormley2020multivariable}. Thanks to the recent developments in the genome-wide association study (GWAS), Mendelian randomization \citep{katan2004commentary} has become a powerful tool to tackle the unmeasured confounding problem \citep{kang2016instrumental,bowden2016consistent,windmeijer2019use,hartwig2017robust,guo2018confidence,zhao2020statistical,ye2021debiased}. In Mendelian randomization analysis, as discussed in Example \ref{ex: invalid iv}, each SNP is used as instrumental variables to estimate the causal effect and the estimator is consistent if the SNP is a valid instrument. The final estimator is obtained via combining the estimators produced by each SNP to improve the efficiency. However, if some SNPs are invalid instruments due to pleiotropy, linkage disequilibrium, and population stratification, the final estimator may be biased. 

In this subsection, we use the comprehensive smoking index (CSI) and alcoholic drinks per week (DPW) as quantitative measures of smoking and alcohol intake and conduct the analysis using the genetic data provided by \cite{gormley2020multivariable}. A copy of the data used in this subsection is available at \texttt{https://
github.com/rcrichmond/smoking\_alcohol\_headandneckcancer}. The data set contains the estimated effect of $168$ independent SNPs on the head and neck cancer, CSI and DPW and the corresponding standard error. Summary-level data for the effect on head and neck cancer is from a GWAS with sample size of $12,619$ conducted by the Genetic Associations and Mechanisms in Oncology Network \citep{lesseur2016genome}. Summary-level data for the effect on CSI is derived by \cite{wootton2020evidence} from the UK BioBank (sample size $462,690$), and the DPW data is obtained from a GWAS with sample size $226,223$ in the GWAS \& Sequencing Consortium of Alcohol and Nicotine use. See \cite{gormley2020multivariable} for further details of the data.
Following \cite{gormley2020multivariable}, we conduct the univariable Mendelian randomization to analyze the causal effect of CSI and DPW separately.
Estimators of the causal effect of  CSI on head and neck cancer are constructed based on $108$ SNPs used in \cite{gormley2020multivariable}, which produce $108$ estimates. The IVW method that uses all these $108$ estimates gives the estimate of $1.791$.
To mitigate the invalid instrument problem, we combine these $108$ estimators by the procedures proposed in this paper, which gives $\tilde{\theta} = 1.956$ and $\hat{\theta} = 1.856$, respectively. 
The results are close to that produced by IVW method. This is in conformity with the fact that, among the $108$ SNPs, no invalid instrument is identified by \cite{gormley2020multivariable}. The analysis result suggests a positive causal effect of CSI on head and neck cancer with confidence interval $[0.969, 2.744]$ (based on $\hat{\theta}$).

Then the causal effect of  DPW is estimated similarly based on $60$ SNPs used in \cite{gormley2020multivariable}. The result of the IVW method is $2.111$. The results of the proposed methods are $\tilde{\theta} = 1.622$ and $\hat{\theta} = 1.598$. When analysing the causal effect of DPW, \cite{gormley2020multivariable} identify an invalid instrument \emph{rs1229984}. When \emph{rs1229984} is not included, the IVW method based on the remaining $59$ SNPs gives the estimate $1.381$, which is quite different from the case where \emph{rs1229984}  is included. The proposed estimators based on the remaining $59$ SNPs are $\tilde{\theta} = 1.619$ and $\hat{\theta} = 1.590$, which are close to the case where \emph{rs1229984}  is included. This demonstrates the robustness of the proposed methods against the invalid instrument \emph{rs1229984}. The analysis result suggests that DPW has a positive causal effect on the head and neck cancer with confidence interval $[0.414,2.774]$  (based on $\hat{\theta}$ without \emph{rs1229984}). We then compare these results with four standard methods in Mendelian randomization problem with invalid instruments, the MR-Egger regression \citep{bowden2015mendelian}, the weighted median method \citep{bowden2016consistent}, the weighted mode method \citep{hartwig2017robust} and the RAPS method \citep{zhao2020statistical}. Results of these four methods are calculated by the R package \emph{TwoSampleMR} (\texttt{https://github.com/MRCIEU/TwoSampleMR}). In the presence of \emph{rs1229984}, the MR-Egger regression, the weighted median, the weighted mode and the RAPS method produce the estimate $2.797$, $2.968$, $2.837$ and $2.165$ for the causal effect of DPW, respectively. Without \emph{rs1229984}, results of these four methods becomes $1.072$, $1.397$, $1.264$ and $1.637$, respectively. All these four standard methods appears to be much more sensitive to the invalid instrument \emph{rs1229984} compared to the proposed methods.

\section{Discussion}
In this paper, we present a fusion-extraction procedure to combine summary statistics from different data sources in the presence of biased sources. 
The idea of the proposed method is quite general and is applicable to many estimation problems. However, several questions are left open by this paper. First, the results in this paper do not apply to the conventional random-effect model in meta-analysis. It warrant further investigation to extend results in this paper to random-effect models. Second, we assume in this paper that different data sources share the same true parameter but some data sources fail to provide a consistent estimator. In practice, the true parameters in different data sources might be heterogeneous and the estimator from a data sources may converge to the true parameter of the data source \citep{claggett2014meta}. In this case, \eqref{eq: identification} defines the ``least false parameter" that minimizes the weighted average distance to true parameters of each data source. It is of interest to investigate the theoretical properties of $\tilde{\theta}$ and $\hat{\theta}$ in this case.

\newpage
	\appendix

	Throughout this proof, we use $B_{L}$ ($B_{\rm U}$) to denote the lower (upper) bound of a positive sequence that is bounded away from zero (infinity).
	\section{Details of the counter example in Section \ref{subsec: estimation}}\label{sec: counter example}
	In this subsection, we prove \eqref{eq: counter ex} in Section \ref{subsec: estimation}.
	\begin{proof}
		Let $n_{*} = n/K$ and $Z_{k} = \tilde{\theta}_{k} - \theta_{0}$. Then it is easy to verify that $\tilde{\theta}={\rm med}(\tilde{\theta}_{1},\dots,\tilde{\theta}_{K})$ and hence $\tilde{\theta} - \theta_{0} = {\rm med}(Z_{1}, \dots, Z_{K})$ where ${\rm med}(\cdot)$ is the univariate median. Let
		\[F_{K}(z) =\frac{1}{K}\sum_{k=1}^{K}1\{Z_{k}\leq z\}. \] 
		Recall that $h_{*} = \Phi^{-1}((3/8 + \tau/4)/(1/2 + \tau))$.
		By the definition of median, to prove \eqref{eq: counter ex}, it suffices to show $F_{K}(h_{*}/\sqrt{n_{*}}) < 1/2$ with probability approaching one. Note that $E[1\{Z_{k}\leq z\}] = \Phi(\sqrt{n_{*}}(z-b_{k}^{*}))$ for any $z$.
		Letting $z = h_{*}/\sqrt{n_{*}}$, according to the Hoeffding inequality \citep[Proposition 2.5]{wainwright2019high}, we have
		\begin{equation}\label{eq: distribution concentration}
		F_{K}(h_{*}/\sqrt{n_{*}}) - \frac{1}{K}\sum_{k=1}^{K}\Phi(h_{*}-b_{k}^{*}\sqrt{n_{*}}) \leq K^{-\frac{1}{4}}
		\end{equation}
		with probability at least $1-2\exp(-2\sqrt{K})$. Clearly, as $K\to \infty$, $1-2\exp(-2\sqrt{K})\to 1$. 
		Notice that
		\begin{equation}\label{eq: expectation bound}
		\begin{aligned}
		\frac{1}{K}\sum_{k=1}^{K}\Phi(h_{*}-b_{k}^{*}\sqrt{n_{*}})
		&= \frac{1}{K}\sum_{k\in \cK_{0}} \Phi(h_{*}) + o(1) \\
		&= \frac{1}{K}\sum_{i=1}^{\lfloor(\frac{1}{2}+\tau)K\rfloor} \frac{3/8 + \tau/4}{1/2 + \tau} + o(1) \\
		&\leq \left(\frac{1}{2} + \tau\right)\frac{3/8 + \tau/4}{1/2 + \tau} + o(1)\\
		&= \frac{3}{8} + \frac{\tau}{4} + o(1).
		\end{aligned}
		\end{equation} Because $3/8 + \tau/4 < 1/2$,
		combining \eqref{eq: distribution concentration} and \eqref{eq: expectation bound}, we have $F_{K}(h_{*}/\sqrt{n_{*}}) < 1/2$ with probability approaching one, which implies \eqref{eq: counter ex}.
	\end{proof}
	\section{Proof of Proposition \ref{prop: identification}}\label{subsec: proof of proposition}
	\begin{proof}
		Let $G(\theta) = \sum_{k=1}^K \tilde{\pi}_{k}\|\theta^{*}_{k} - \theta\| = \sum_{k \in \cK_{0}}\tilde{\pi}_{k}\|\theta_{0} - \theta\| + \sum_{k \in \cK_{0}^{c}}\tilde{\pi}_{k}\|\theta_{0} + b^{*}_{k} - \theta\|$. For any $\theta^{\prime}\not= \theta_{0}$, the directional derivative of $G(\theta)$ at the point $\theta_{0}$ in the direction $\theta^{\prime} - \theta_{0}$ is 
		\begin{align*}
		&\|\theta^{\prime} - \theta_{0}\|\left(\sum_{k \in \cK_{0}}\tilde{\pi}_{k} + \sum_{k \in \cK_{0}^{c}}\tilde{\pi}_{k}\frac{b^{*\T}_{k}(\theta^{\prime} - \theta_{0})}{\|b^{*}_{k}\|\|\theta^{\prime} - \theta_{0}\|}\right) \\
		& = \|\theta^{\prime} - \theta_{0}\|\left\{\sum_{k \in \cK_{0}}\tilde{\pi}_{k} + \left(\sum_{k \in \cK_{0}^{c}}\tilde{\pi}_{k}\frac{b^{*}_{k}}{\|b^{*}_{k}\|}\right)^\T\frac{\theta^{\prime} - \theta_{0}}{\|\theta^{\prime} - \theta_{0}\|}\right\}\\
		& \geq \|\theta^{\prime} - \theta_{0}\|\left(\sum_{k \in \cK_{0}}\tilde{\pi}_{k} - \left\|\sum_{k \in \cK_{0}^{c}}\tilde{\pi}_{k}\frac{b^{*}_{k}}{\|b^{*}_{k}\|}\right\|\right) > 0.
		\end{align*}
		Hence Proposition \ref{prop: identification} follows from the fact that $G(\theta)$ is convex. 
		
	\end{proof}
	\section{Proof of Theorem \ref{thm: consistent}}\label{subsec: proof of thm consistent}
	Theorem \ref{thm: consistent} is a straightforward corollary of the following Lemma.
	\begin{lemma}\label{lem: consistent}
		If $\delta > 0$, then
		\[
		\|\tilde{\theta} - \theta_{0}\| \leq 2\delta^{-1}\sum_{k=1}^{K}\tilde{\pi}_{k}\|\tilde{\theta}_{k} - \theta^{*}_{k}\|. 
		\]
	\end{lemma}
	\begin{proof}
		By the convexity of $G(\theta)$ and the directional derivative given in the proof of Proposition \ref{prop: identification}, we have
		\begin{equation}\label{eq: th1-1}
		G(\theta) - G(\theta_{0}) \geq \delta\|\theta - \theta_{0}\|.
		\end{equation}
		Let $\tilde{G}(\theta) = \sum_{k=1}^K \tilde{\pi}_{k}\|\tilde{\theta}_{k} - \theta\|$. Then by the triangle inequality of the Euclid norm,
		\[|\tilde{G}(\theta) - G(\theta)| \leq \sum_{k=1}^{K}\tilde{\pi}_{k}\|\tilde{\theta}_{k} - \theta^{*}_{k}\|\]
		for any $\theta$. Thus
		\[\tilde{G}(\theta) - \tilde{G}(\theta_{0}) \geq G(\theta) - G(\theta_{0}) - 2\sum_{k=1}^{K}\tilde{\pi}_{k}\|\tilde{\theta}_{k} - \theta^{*}_{k}\|.\]
		This together with \eqref{eq: th1-1} proves
		\[\tilde{G}(\theta) - \tilde{G}(\theta_{0}) > 0\]
		for all $\theta$ satisfying $\|\theta - \theta_{0}\| > 2\delta^{-1}\sum_{k=1}^{K}\tilde{\pi}_{k}\|\tilde{\theta}_{k} - \theta^{*}_{k}\|$. Recalling the definition of $\tilde{\theta}$ in Section \ref{subsec: estimation}, we have $\tilde{G}(\tilde{\theta}) \leq \tilde{G}(\theta_{0})$ and hence $\|\tilde{\theta} - \theta_{0}\| \leq 2\delta^{-1}\sum_{k=1}^{K}\tilde{\pi}_{k}\|\tilde{\theta}_{k} - \theta^{*}_{k}\|$. 
		
	\end{proof}
	
	\section{Proof of Theorem \ref{thm: equivalence}}\label{subsec: proof of thm equivalence}
	To prove Theorem \ref{thm: equivalence}, we first establish two useful lemmas. Here we just state the lemmas and the key ideas. See the next Section for the formal proof of the two lemmas are relegated.
	
	A key step of the proof is to construct a ``good event" that happens with high probability and on the good event $\hat{\theta}$ has some desirable properties.
	
	For any positive numbers $\epsilon_{n}$ and constants $C_{\rm L}>0$ and $C_{\rm U}>1$ such that $C_{\rm L} < B_{\rm L}  \leq B_{\rm U} < C_{\rm U}$, let $\Delta_{\rm M} = \min\{B_{\rm L} - C_{\rm L}, C_{\rm U} - B_{\rm U}\}$. We first construct three event as follows,
	\begin{equation}\label{def: proper set}
	\begin{aligned}
	&\cS_{1} = \left\{\|\tilde{V}_{k} - V_{k}^{*}\| \leq \Delta_{\rm M}, \ k\in \cK_{0}\right\}, \\
	&\cS_{2} = \left\{\|\hat{\theta}_{\rm IVW} - \tilde{\theta}_{k}\| < (2\tilde{\pi}_{k}C_{\rm U})^{-1}\lambda\tilde{w}_{k}, \ k\in \cK_{0}\right\},\\ 
	& \cS_{3} = \left\{\min_{k\in \cK_{0}}\frac{\lambda\tilde{w}_{k}}{\tilde{\pi}_{k}} \geq 2\epsilon_{n},\ C_{\rm U}\sum_{k\in \cK_0^{c}}\lambda\tilde{w}_{k}\leq  B_{\rm L}C_{\rm L}\epsilon_{n}\right\}.
	\end{aligned}
	\end{equation}
	On $\cS_{1}$, for $k\in\cK_{0}$, $\tilde{V}_{k}$ is close to $V_{k}^{*}$. On $\cS_{2}$, for $k\in \cK_{0}$, the penalty coefficient of $b_{k}$ in problem \eqref{eq: opt problem} dominates the difference between $\hat{\theta}_{\rm IVW}$ and $\tilde{\theta}_{k}$. Hence $b_{k}$'s are likely to be penalized to zero for $k\in \cK_{0}$ on $\cS_{2}$. On $\cS_{3}$, the penalty coefficient of $b_{k}$ is not too small for $k\in\cK_{0}$ and not too large for $k\in \cK_{0}^{c}$. Intuitively, these three events are all good events on which $\hat{\theta}$ would perform well. Let $\cS = \cS_{1}\cap\cS_{2}\cap\cS_{3}$. Next, we show that $\hat{\theta}$ is close to $\hat{\theta}_{\rm IVW}$ on the event
	$\cS$. The formal result is stated in the following lemma.
	\begin{lemma}\label{lem: close set}
		On the event $\cS$, we have
		\[\|\hat{\theta} - \hat{\theta}_{\rm IVW}\| \leq \epsilon_{n}.\]
	\end{lemma}
    Under Conditions \ref{cond: existence of nonbias} and \ref{cond: matrix rate} and some conditions on the convergence rate of $\|\tilde{\theta}_{k} - \theta_{k}^{*}\|$, we have $P(\cS) \to 1$ and hence $P(\hat{\theta} - \hat{\theta}_{\rm IVW}\| \leq \epsilon_{n}) \to 1$ according to Lemma \ref{lem: close set}. The formal result is summarized in the following lemma. See Appendix \ref{app: proofs on lemmas} for the proof of Lemma \ref{lem: close set} and \ref{lem: main}.
	\begin{lemma}\label{lem: main}
		Under Conditions \ref{cond: existence of nonbias} and \ref{cond: matrix rate}, if the tuning parameter $\lambda$ satisfies 
	    \[
		\lambda^{-1}\delta^{-\alpha}\max_{k}\{\|\tilde{\theta}_{k} - \theta_{k}^{*}\|^{\alpha + 1}\} = o_{P}(1),
		\]
		then for any sequence $\epsilon_{n}$ such that $\lambda K/\epsilon_{n} \to 0$ and $\epsilon_{n}\lambda^{-1}\delta^{-\alpha}\max_{k}\{\|\tilde{\theta}_{k} - \theta_{k}^{*}\|^{\alpha}\} = o_{P}(1)$, 
		\[P(\|\hat{\theta} - \hat{\theta}_{\rm IVW}\| \leq \epsilon_{n}) \to 1.\]
	\end{lemma}
	
	With the assistance of Lemma \ref{lem: main}, we are able to prove Theorem \ref{thm: equivalence}.
	\begin{proof}
		Condition \ref{cond: estimate rate} and the fact that $\alpha > \max\{\nu_{1}\nu_{2}^{-1}, \nu_{2}^{-1} - 1\}$ together imply
		\begin{equation}\label{eq: estimate rate}
		\begin{aligned}
		&\delta^{-(\alpha + 1)}\max_{k}\{\|\tilde{\theta}_{k} - \theta_{k}^{*}\|^{\alpha + 1}\} = o_{P}(n^{-1}),\\ &\delta^{-\alpha}\max_{k}\{\|\tilde{\theta}_{k} - \theta_{k}^{*}\|^{\alpha}\} = o_{P}(n^{-\alpha \nu_{2}}).
		\end{aligned}
		\end{equation}
		Because  $\lambda \asymp 1/n$, the conditions Lemma \ref{lem: main} is satisfied with 	$\epsilon_{n} = a_{n}K/n$ where $a_{n}$ is an arbitrary sequence of positive numbers such that $a_{n}\to\infty$ and $a_{n}n^{\nu_{1} - \alpha\nu_{2}}\to0$. Note that $\nu_{1} - \alpha\nu_{2} < 0$. Then we have 
		\begin{equation}\label{eq: contradiction}
			P(\|\hat{\theta} - \hat{\theta}_{\rm IVW}\|\leq a_{n}K/n) \to 1
		\end{equation} for arbitrary $a_n$ that diverges to infinity at a sufficiently slow rate. This indicates that $\|\hat{\theta} - \hat{\theta}_{\rm IVW}\| = O_{P}(K/n)$. To see this, assuming that $n\|\hat{\theta} - \hat{\theta}_{\rm IVW}\|/K$ is not bounded in probability, then for some $\epsilon > 0$ there is some $m_{1} \geq e$ such that $P(n\|\hat{\theta} - \hat{\theta}_{\rm IVW}\|/K > 1) \geq \epsilon$ when $n = m_{1}$. For $s=2,3,\dots$, there is some $m_{s} > \max\{m_{s-1}, e^{s}\}$ such that $P(n\|\hat{\theta} - \hat{\theta}_{\rm IVW}\|/K > s) \geq \epsilon$ when $n = m_{s}$. Let $a_{n} = s$ for $m_{s} \leq n < m_{s + 1}$. Then for this sequence, we have $a_{n}\to \infty$ and $a_{n} \leq \log n$. Hence $a_{n}$ satisfies $a_{n}n^{\nu_{1} - \alpha\nu_{2}}\to0$. Moreover, for any positive integer $s$, $P(\|\hat{\theta} - \hat{\theta}_{\rm IVW}\| \leq a_{n}K/n) \leq 1-\epsilon$ when $n=m_{s}$. Thus, $\liminf_{n} P(\|\hat{\theta} - \hat{\theta}_{\rm IVW}\| \leq a_{n}K/n) \leq 1 - \epsilon$, which contradicts to \eqref{eq: contradiction}.
	\end{proof}
\section{Proof of Lemmas \ref{lem: close set} and \ref{lem: main}}\label{app: proofs on lemmas}
To prove the two lemmas, we first analyse the optimization problem \eqref{eq: opt problem} in the main text.
We denote $\gamma = (\theta^\T, b^\T_{1},\dots,b^\T_{K})^\T$ as a grand parameter vector.
Let  
\[\Gamma_{0}= \{\gamma:\gamma = (\theta, b^\T_{1},\dots,b^\T_{K})^\T,\ b_{k} = 0\ \text{for} \ k\in \cK_{0}\ \text{and} \ b_{k} \not= 0 \ \text{for} \ k \in \cK_{0}^{c}\},\]
and 
\[L(\gamma) = \sum_{k=1}^{K}\frac{\tilde{\pi}_{k}}{2}(\tilde{\theta}_{k} - \theta - b_{k})^\T\tilde{V}_{k}(\tilde{\theta}_{k} - \theta - b_{k}).\]
Consider the following oracle problem that sets the term $b_k$ to be zero in prior for $k \in \cK_{0}$
\begin{equation}\label{eq: oracle problem}
\min_{\gamma \in \Gamma_{0}}\{L(\gamma) + \sum_{k\in\cK_{0}^{c}}\lambda\tilde{w}_{k}\|b_{k}\|\}.
\end{equation}
The following lemma establishes the relationship between the minimum point of this problem and the problem \eqref{eq: opt problem}.
\begin{lemma}\label{lem: equivalent}
	Let $\cM$ be the set of minimum points of problem \eqref{eq: opt problem} and $\bar{\cM}$ be the set of minimum points of problem \eqref{eq: oracle problem}.
	If there exists a minimum point $\bar{\gamma} = (\bar{\theta}^\T,\bar{b}^\T_{1},\dots,\bar{b}^\T_{K})^\T$ of problem \eqref{eq: oracle problem} such that $\tilde{\pi}_{k}\|\tilde{V}_{k}(\tilde{\theta}_{k} - \bar{\theta})\| < \lambda\tilde{w}_{k}$ for $k\in \cK_{0}$, then $\cM = \bar{\cM}$.
\end{lemma}
\begin{proof}
	Because $\bar{\gamma}$ is a minimum point of problem \eqref{eq: oracle problem}, it follows from the Karush-Kunh-Tucker condition that
	\begin{equation}\label{eq: oracle KKT}
	\left\{\begin{array}{l}
	\sum_{k\in \cK_{0}}\tilde{\pi}_{k} \tilde{V}_{k}(\tilde{\theta}_{k} - \bar{\theta}) + \sum_{k\in \cK_{0}^{c}}\tilde{\pi}_{k} \tilde{V}_{k}(\tilde{\theta}_{k} - \bar{\theta} - \bar{b}_{k}) = 0, \\
	\tilde{\pi}_{k}\tilde{V}_{k}(\tilde{\theta}_{k} - \bar{\theta} - \bar{b}_{k}) = \lambda\tilde{w}_{k}\bar{z}_{k}, \ k \in \cK_{0}^{c}
	\end{array}\right.
	\end{equation}
	where $\bar{z}_{k} = \bar{b}_{k}/\|\bar{b}_{k}\|$ if $\bar{b}_{k} \not = 0$ and $\|\bar{z}_{k}\| \leq  1$ if $\bar{b}_{k} = 0$. Because $\tilde{\pi}_{k}\|\tilde{V}_{k}(\tilde{\theta}_{k} - \bar{\theta})\| < \lambda\tilde{w}_{k}$ for $k \in \mathcal{K}_{0}$, $\bar{\gamma}$ also satisfies the Karush-Kunh-Tucker condition of problem \eqref{eq: opt problem} and hence $\bar{\gamma}\in \cM$ by the convexity of problem \eqref{eq: opt problem}. 
	
	One the one hand, for any $\bar{\gamma}^{\prime}\in \bar{\cM}$, because both $\bar{\gamma}^{\prime}$ and $\bar{\gamma}$ belongs to $\bar{\cM}$, we have $L(\bar{\gamma}) + \sum_{k=1}^{K}\lambda\tilde{w}_{k}\|\bar{b}_{k}\| = L(\bar{\gamma}^{\prime}) + \sum_{k=1}^{K}\lambda\tilde{w}_{k}\|\bar{b}^{\prime}_{k}\|$. In addition, according to the above discussion, we have $\bar{\gamma} \in \cM$. Then we have $L(\bar{\gamma}) + \sum_{k=1}^{K}\lambda\tilde{w}_{k}\|\bar{b}_{k}\| = \min_{\gamma}\{L(\gamma) + \sum_{k=1}^{K}\lambda\tilde{w}_{k}\|b_{k}\|\}$ and hence $L(\bar{\gamma}^{\prime}) + \sum_{k=1}^{K}\lambda\tilde{w}_{k}\|\bar{b}^{\prime}_{k}\| = \min_{\gamma}\{L(\gamma) + \sum_{k=1}^{K}\lambda\tilde{w}_{k}\|b_{k}\|\}$. This implies $\bar{\gamma}^{\prime}\in \cM$ and proves $\bar{\cM}\subset\cM$. 
	
	On the other hand, for any $\gamma^{\prime}\in \cM$,  because both $\gamma^{\prime}$ and $\bar{\gamma}$ belongs to $\bar{\cM}$, we have $L(\bar{\gamma}) + \sum_{k=1}^{K}\lambda\tilde{w}_{k}\|\bar{b}_{k}\| = L(\gamma^{\prime}) + \sum_{k=1}^{K}\lambda\tilde{w}_{k}\|b^{\prime}_{k}\|$, this implies $L(\bar{\gamma}) - L(\gamma^{\prime}) = \sum_{k=1}^{K}\lambda\tilde{w}_{k}(\|b^{\prime}_{k}\| - \|\bar{b}_{k}\|)$. Recall that, for $k \in \mathcal{K}_{0}^{c}$,  $\bar{z}_{k} = \bar{b}_{k}/\|\bar{b}_{k}\|$ if $\bar{b}_{k} \not = 0$ and $\|\bar{z}_{k}\| \leq  1$ if $\bar{b}_{k} = 0$. In addition, let $\bar{z}_{k} = (\lambda\tilde{w}_{k})^{-1}\tilde{\pi}_{k}\tilde{V}_{k}(\tilde{\theta}_{k} - \bar{\theta})$
	for $k\in \cK_{0}$.  Then it is easy to verify that $\|\bar{z}_{k}\| \leq 1$ for $k \in \cK_{0}^{c}$, $\|\bar{z}_{k}\| < 1$ for $k\in \cK_{0}$, $\bar{z}^\T_{k}\bar{b}_{k} = \|\bar{b}_{k}\|$ for $k = 1,\dots,K$, and $\nabla L(\bar{\gamma}) = - (0^\T,  \lambda\tilde{w}_1\bar{z}^\T_1,\dots, \lambda\tilde{w}_K\bar{z}^\T_K)^\T$. By the convexity of $L(\gamma)$, we have
	\begin{align*}
	0 & \geq L(\bar{\gamma}) - L(\gamma^{\prime}) + \nabla L(\bar{\gamma})^\T(\gamma^{\prime} - \bar{\gamma})\\
	& = \sum_{k=1}^{K}\lambda\tilde{w}_{k}(\|b^{\prime}_{k}\| - \|\bar{b}_{k}\|) - \sum_{k=1}^{K}\lambda\tilde{w}_{k}\bar{z}^{\T}_{k}(b^{\prime}_{k} - \bar{b}_{k})\\
	& = \sum_{k=1}^{K}\lambda\tilde{w}_{k}(\|b^{\prime}_{k}\| - \bar{z}_{k}^\T b^{\prime}_{k}) \\
	& \geq \sum_{k=1}^{K}\lambda\tilde{w}_{k}(\|b^{\prime}_{k}\| - \|\bar{z}_{k}\|\|b^{\prime}_{k}\|) \\
	& \geq \sum_{k \in \cK_{0}}\lambda\tilde{w}_{k}\|b^{\prime}_{k}\|(1 - \|\bar{z}_{k}\|).
	\end{align*}
	Because $\lambda\tilde{w}_{k}(1 - \|\bar{z}_{k}\|) = \lambda\tilde{w}_{k} - \tilde{\pi}_{k}\|\tilde{V}_{k}(\tilde{\theta}_{k} - \bar{\theta})\| > 0$ for $k \in \cK_{0}$, we have $\|b^{\prime}_{k}\| = 0$ for $k\in \cK_{0}$. Combining this with the fact that $L(\bar{\gamma}) + \sum_{k=1}^{K}\lambda\tilde{w}_{k}\|\bar{b}_{k}\| = L(\gamma^{\prime}) + \sum_{k=1}^{K}\lambda\tilde{w}_{k}\|b^{\prime}_{k}\|$, we have $\gamma^{\prime}\in \bar{\cM}$. This completes the proof of the lemma.
	
\end{proof}

Recall that throughout the proofs, we always use $B_{L}$ ($B_{\rm U}$) to denote the lower (upper) bound of a sequence that is bounded away from zero (infinity).	For any positive numbers $\epsilon_{n}$ and constants $C_{\rm L} >0$ and $C_{\rm U} > 1$ such that $C_{\rm L} < B_{\rm L}  \leq B_{\rm U} < C_{\rm U}$, let $\Delta_{\rm M} = \min\{B_{\rm L} - C_{\rm L}, C_{\rm U} - B_{\rm U}\}$,
\begin{equation*}
\begin{aligned}
&\cS_{1} = \{\|\tilde{V}_{k} - V_{k}^{*}\| \leq \Delta_{\rm M}, \ k\in \cK_{0}\}, \\
&\cS_{2} = \{\|\hat{\theta}_{\rm IVW} - \tilde{\theta}_{k}\| < (2\tilde{\pi}_{k}C_{\rm U})^{-1}\lambda\tilde{w}_{k}, \ k\in \cK_{0}\},\\ 
& \cS_{3} = \{\min_{k\in \cK_{0}}\lambda\tilde{w}_{k}/\tilde{\pi}_{k} \geq 2\epsilon_{n},\ C_{\rm U}\sum_{k\in \cK_0^{c}}\lambda\tilde{w}_{k}\leq  B_{\rm L}C_{\rm L}\epsilon_{n}\},\\
&\text{and}\\
& \cS = \cS_{1}\cap\cS_{2}\cap\cS_{3}.
\end{aligned}
\end{equation*}
Then we are ready to give the proof of Lemma \ref{lem: close set}.
\paragraph{\bf Restate of Lemma \ref{lem: close set}.} 
\emph{On the event $\cS$, we have
	\[\|\hat{\theta} - \hat{\theta}_{\rm IVW}\| \leq \epsilon_{n}.\]
}
\begin{proof}
	Let $\bar{\gamma}\in \bar{\cM}$ be any minimum point of problem \eqref{eq: oracle problem}. Then by the KKT condition \eqref{eq: oracle KKT}, we have
	\begin{equation}\label{eq: oracle expansion}
	\begin{aligned}
	\bar{\theta}& = (\sum_{k \in \cK_{0}}\tilde{\pi}_{k}\tilde{V}_{k})^{-1}(\sum_{k \in \cK_{0}}\tilde{\pi}_{k}\tilde{V}_{k}\tilde{\theta}_{k}) + 
	(\sum_{k \in \cK_{0}}\tilde{\pi}_{k}\tilde{V}_{k})^{-1}(\sum_{k\in \cK_{0}^{c}}\lambda\tilde{w}_{k}\bar{z}_{k})\\
	&= \hat{\theta}_{\rm IVW} + \left(\sum_{k \in \cK_{0}}\tilde{\pi}_{k}\tilde{V}_{k}\right)^{-1}\left(\sum_{k\in \cK_{0}^{c}}\lambda\tilde{w}_{k}\bar{z}_{k}\right).
	\end{aligned}
	\end{equation}
	By Weyl's Theorem, we have $\max_{k\in \cK_{0}}\{\max\{|\lambda_{\rm min}(\tilde{V}_{k}) - \lambda_{\rm min}(V_{k}^{*})|, |\lambda_{\rm max}(\tilde{V}_{k}) - \lambda_{\rm max}(V_{k}^{*})|\}\} \leq \max_{k\in \cK_{0}}\|\tilde{V}_{k} - V_{k}^{*}\| \leq \Delta_{\rm M}$ on $\cS_{1}$.
	Then by Condition \ref{cond: matrix rate}, it follows 
	\[C_{\rm L} \leq \lambda_{\rm min}(\tilde{V}_{k}) \leq \lambda_{\rm max}(\tilde{V}_{k}) \leq C_{\rm U}\]
	for $k \in \cK_{0}$ on $\cS_{1}$.
	Then  by Conditions \ref{cond: existence of nonbias} and \ref{cond: matrix rate}, we have
	\begin{align*}
	\tilde{\pi}_{k}\|\tilde{V}_{k}(\tilde{\theta}_{k} - \bar{\theta})\|
	& \leq \tilde{\pi}_{k}\|\tilde{V}_{k}\|\|\tilde{\theta}_{k} - \hat{\theta}_{\rm IVW}\| + \tilde{\pi}_{k}\|\tilde{V}_{k}\|\|(\sum_{k \in \cK_{0}}\tilde{\pi}_{k}\tilde{V}_{k})^{-1}\|\|\sum_{k\in \cK_{0}^{c}}\lambda\tilde{w}_{k}\bar{z}_{k}\| \\
	& \leq \tilde{\pi}_{k}C_{\rm U}\|\tilde{\theta}_{k} - \hat{\theta}_{\rm IVW}\| + \tilde{\pi}_{k}C_{\rm U}B_{L}^{-1}C_{\rm L}^{-1}\sum_{k\in \cK_0^{c}}\lambda\tilde{w}_{k} \\
	& <  \frac{\lambda\tilde{w}_{k}}{2} + \tilde{\pi}_{k}\epsilon_{n}\\
	& < \lambda\tilde{w}_{k}
	\end{align*}
	on the event $\cS$.
	According to Lemma \ref{lem: equivalent}, we have $\bar{\cM} = \cM$ on $\cS$. By equation \eqref{eq: oracle expansion}, for any $(\hat{\theta}^{\T},\hat{b}_{1}^{\T},\dots,\hat{b}_{K}^{\T})^{\T} \in \cM = \bar{\cM}$,
	we have
	\begin{equation*}
	\begin{aligned}
	\hat{\theta} - \hat{\theta}_{\rm IVW} = \left(\sum_{k \in \cK_{0}}\tilde{\pi}_{k}\tilde{V}_{k}\right)^{-1}\left(\sum_{k\in \cK_{0}^{c}}\lambda\tilde{w}_{k}\bar{z}_{k}\right).
	\end{aligned}
	\end{equation*}
	This implies
	\[\|\hat{\theta} - \hat{\theta}_{\rm IVW}\| \leq B_{\rm L}^{-1}C_{\rm L}^{-1}\sum_{k\in \cK_0^{c}}\lambda\tilde{w}_{k} \leq C_{\rm U}^{-1}\epsilon_{n}\]
	on $\cS$. Note that $C_{\rm U}> 1$ and this completes the proof of the lemma.
	
\end{proof}

Next, we move on to the proof of Lemma \ref{lem: main}.
\paragraph{\bf Restate of Lemma \ref{lem: main}.} 
\emph{	Under Conditions \ref{cond: existence of nonbias}, and \ref{cond: matrix rate}, if the tuning parameter $\lambda$ satisfies \[\lambda^{-1}\delta^{-\alpha}\max_{k}\{\|\tilde{\theta}_{k} - \theta_{k}^{*}\|^{\alpha + 1}\} = o_{P}(1),\]
	then for any sequence $\epsilon_{n}$ such that $\lambda K/\epsilon_{n} \to 0$ and $\epsilon_{n}\lambda^{-1}\delta^{-\alpha}\max_{k}\{\|\tilde{\theta}_{k} - \theta_{k}^{*}\|^{\alpha}\} = o_{P}(1)$, 
	\[P(\|\hat{\theta} - \hat{\theta}_{\rm IVW}\| \leq \epsilon_{n}) \to 1.\]
}
\begin{proof}
	To prove the lemma, according to Lemma \ref{lem: close set}, it suffices to prove $P(\cS) \to 1$ with $C_{\rm L} = 0.9B_{\rm L}$, $C_{\rm U} = 1.1B_{U}$. By Condition \ref{cond: matrix rate}, we have $P(\cS_{1}) \to 1$.
	By the definition of $\hat{\theta}_{\rm IVW}$, we have
	\[\hat{\theta}_{\rm IVW} - \tilde{\theta}_{k} = \left(\sum_{j \in \cK_{0}}\tilde{\pi}_{j}\tilde{V}_{j}\right)^{-1}\left(\sum_{j \in \cK_{0}}\tilde{\pi}_{j}\tilde{V}_{j}(\tilde{\theta}_{j} - \tilde{\theta}_{k})\right).\]
	For $k \in \cK_{0}$, on the event $\cS_{1}$, we have
	\begin{equation}\label{eq: bound difference IVW-k}
	\begin{aligned}
	\|\hat{\theta}_{\rm IVW} - \tilde{\theta}_{k}\|& \leq (B_{\rm L} - \max_{j\in \cK_{0}}\|\tilde{V}_{j} - V_{j}^{*}\|)^{-1}(\max_{j\in \cK_{0}}\|\tilde{V}_{j}(\tilde{\theta}_{j} - \tilde{\theta}_{k})\|)\\
	& \leq (B_{\rm L} - \max_{j\in \cK_{0}}\|\tilde{V}_{j} - V_{j}^{*}\|)^{-1}(B_{\rm U} + \max_{j\in \cK_{0}}\|\tilde{V}_{j} - V_{j}^{*}\|)\max_{j\in \cK_{0}}\|\tilde{\theta}_{j} - \tilde{\theta}_{k}\|\\
	& \leq C_{\rm L}^{-1}C_{\rm U}\max_{j\in \cK_{0}}\|\tilde{\theta}_{j} - \tilde{\theta}_{k}\|
	\end{aligned}
	\end{equation}
	by Condition \ref{cond: matrix rate}.
	Because for $j \in \cK_{0}$, $\theta_{j}^{*} = \theta_{0}$,
	we have
	\begin{equation}\label{eq: bound difference j-k}
	\max_{j\in \cK_{0}}\|\tilde{\theta}_{j} - \tilde{\theta}_{k}\| \leq  \max_{j}\|\tilde{\theta}_{j} - \theta_{0}\| + \max_{k}\|\tilde{\theta}_{k} - \theta_{0}\| = 2\max_{j}\|\tilde{\theta}_{j} - \theta_{j}^{*}\|.
	\end{equation}
	Note that by Lemma \ref{lem: consistent},
	\[\|\tilde{\theta} - \theta_{0}\| \leq 2\delta^{-1} \max_{j}\|\tilde{\theta}_{j} - \theta_{j}^{*}\|.\]
	This together with the definitions of $\tilde{b}_{k}$ and $b_{k}^{*}$ proves
	\begin{equation}\label{eq: bias est}
	\|\tilde{b}_{k} - b_{k}^{*}\| \leq (1 + 2\delta^{-1}) \max_{j}\|\tilde{\theta}_{j} - \theta_{j}^{*}\|.
	\end{equation}
	Recalling that $\tilde{w}_{k} = 1 / \|\tilde{b}_{k}\|^{\alpha}$, according to \eqref{eq: bound difference IVW-k}, \eqref{eq: bound difference j-k} and \eqref{eq: bias est} we have
	\begin{align*}
	\cS_{1}\cap\cS_{2}& = \cS_{1}\cap\{\lambda^{-1}\tilde{\pi}_{k}\|\tilde{b}_{k}\|^{\alpha}\|\hat{\theta}_{\rm IVW} - \tilde{\theta}_{k}\| < (2C_{\rm U})^{-1}, \ k\in \cK_{0}\}\\
	& \supset\cS_{1}\cap\{2C_{\rm L}^{-1}C_{\rm U}\lambda^{-1}(1+2\delta^{-1})^{\alpha}(\max_{j}\|\tilde{\theta}_{j} - \theta_{j}^{*}\|)^{\alpha + 1} 
	< (2C_{\rm U})^{-1}\}\\
	& \equalscolon \cS_{1}\cap\cS_{2}^{*}.
	\end{align*} 
	Because $\lambda^{-1}\delta^{-\alpha}\max_{j}\{\|\tilde{\theta}_{j} - \theta_{j}^{*}\|^{\alpha + 1}\} = o_{P}(1)$, then $P(\cS_{2}^{*}) \to 1$ and hence $P(\cS_{1}\cap\cS_{2}^{*}) \to 1$. This implies $P(\cS_{1}\cap\cS_{2}) \to 1$.
	Notice that $\cS_{3}$ can be rewritten as 
	\[
	\{2\epsilon_{n}\lambda^{-1}\max_{k\in \mathcal{K}_{0}}\{\tilde{\pi}_{k}\|\tilde{b}_{k}\|^{\alpha}\} \leq 1,  \sum_{k\in \cK_0^{c}}\lambda\tilde{w}_{k}/\epsilon_{n}\leq C_{\rm U}^{-1}B_{\rm L}C_{\rm L}\}.
	\]
	According to \eqref{eq: bias est}, because $\|b_{k}^{*}\|$ is bounded away from zero for $k\in \cK_{0}^{c}$,
	\begin{align*}
	\cS_{3} &\supset \{2\epsilon_{n}\lambda^{-1} (1 + 2\delta^{-1})^{\alpha} \max_{k}\{\|\tilde{\theta}_{k} - \theta_{k}^{*}\|^{\alpha}\}\leq 1,\ \sum_{k\in \cK_0^{c}}\lambda\tilde{w}_{k}/\epsilon_{n}\leq C_{\rm U}^{-1}B_{\rm L}C_{\rm L}\} \\
	& \supset \{2\epsilon_{n}\lambda^{-1} (1 + 2\delta^{-1})^{\alpha} \max_{k}\{\|\tilde{\theta}_{k} - \theta_{k}^{*}\|^{\alpha}\}\leq 1,\ \max_{k\in \cK_{0}^{c}}\|\tilde{w}_{k}\|\lambda K/\epsilon_{n}\leq C_{\rm U}^{-1}B_{\rm L}C_{\rm L}\}\\
	& = \{2\epsilon_{n}\lambda^{-1} (1 + 2\delta^{-1})^{\alpha} \max_{k}\{\|\tilde{\theta}_{k} - \theta_{k}^{*}\|^{\alpha}\}\leq 1,\ \lambda^{-1}\epsilon_{n}\min_{k\in \cK_{0}^{c}}\|\tilde{b}_{k}\|^{\alpha}\geq C_{\rm U}B_{L}^{-1}C_{\rm L}^{-1} K \}\\
	& \supset \{2\epsilon_{n}\lambda^{-1} (1 + 2\delta^{-1})^{\alpha} \max_{k}\{\|\tilde{\theta}_{k} - \theta_{k}^{*}\|^{\alpha}\}\leq 1,\\
	& \phantom{\quad \supset } (B_{\rm L} - (1 + 2\delta^{-1}) \max_{k}\|\tilde{\theta}_{k} - \theta_{k}^{*}\|)^{\alpha}\geq C_{\rm U}B_{L}^{-1}C_{\rm L}^{-1} \lambda K / \epsilon_{n}\} \\
	& \equalscolon \cS_{3}^{*}.
	\end{align*}
	Since $\lambda K / \epsilon_{n} \to 0$, we have $\lambda/\epsilon_{n} \to 0$. Because
	$\epsilon_{n}\lambda^{-1}\delta^{-\alpha}\max_{k}\{\|\tilde{\theta}_{k} - \theta_{k}^{*}\|^{\alpha}\} = o_{P}(1)$, then we have
	\[(1 + 2\delta^{-1}) \max_{k}\|\tilde{\theta}_{k} - \theta_{k}^{*}\| =  (\epsilon_{n}^{-1}\lambda \times \epsilon_{n}\lambda^{-1}(1 + 2\delta^{-1})^{\alpha}\max_{j}\{\|\tilde{\theta}_{j} - \theta_{j}^{*}\|^{\alpha}\})^{\frac{1}{\alpha}} = o_{P}(1).\]
	Hence $P(\cS_{3}^{*}) \to 1$  and hence $P(\cS_{3}) \to 1$. This completes the proof.
	
\end{proof}

\section{Proof of Theorem \ref{thm: selection consistency}}
We first establish a lemma that is needed in the proof of Theorem \ref{thm: selection consistency}.
\begin{lemma}\label{lem: sure screening}
	On the event $\cS$, we have $\hat{b}_{k} = 0$ for $k \in \cK_{0}$. 
\end{lemma}
\begin{proof}
	According to the proof of Lemma \ref{lem: close set}, we have $\bar{\cM} = \cM$ on $\cS$. By the definition of $\bar{\cM}$, for any $(\hat{\theta}^{\T},\hat{b}_{1}^{\T},\dots,\hat{b}_{K}^{\T})^{\T} \in \cM = \bar{\cM}$, we have $\hat{b}_{k} = 0$ for $k \in \cK_{0}$. 
\end{proof}

\paragraph{\bf Restate of Theorem \ref{thm: selection consistency}.}
\emph{Under Conditions \ref{cond: existence of nonbias}, \ref{cond: estimate rate} and \ref{cond: matrix rate 2}, if $\lambda \asymp 1/n$ and $\alpha > \max\{\nu_{1}\nu_{2}^{-1}, \nu_{2}^{-1} - 1\}$, we have
	\[P(\hat{\cK}_{0} = \cK_{0}) \to 1\]
	provided $\min_{k\in \cK_{0}^{c}}\tilde{\pi}_{k} > C_{\pi} /K$ and $K\log n /n \to 0$ where $C_{\pi}$ is some positive constant.}
\begin{proof}
	As before, we let $C_{\rm L} = 0.9B_{\rm L}$, $C_{\rm U} = 1.1B_{U}$, $\Delta_{\rm M} = \min\{B_{\rm L} - C_{\rm L}, C_{\rm U} - B_{\rm U}\}$ and $\epsilon_{n} = a_{n}K/n$ where $a_{n}$ is a sequence of positive numbers such that $a_{n}\to\infty$ and $a_{n} / \log n\to 0$. 
	Define
	\begin{equation*}
	\begin{aligned}
	&\cS_{1}^{\prime} = \{\|\tilde{V}_{k} - V_{k}^{*}\| \leq \Delta_{\rm M}, \ k=1,\dots,K\}, \\
	&\cS_{2}^{\prime} = \{\|\hat{\theta}_{\rm IVW} - \tilde{\theta}_{k}\| > (\tilde{\pi}_{k}C_{\rm L})^{-1}\lambda\tilde{w}_{k} + \epsilon_{n}, \ k\in \cK_{0}^{c}\},\\ 
	& \cS^{\prime} = \cS_{1}^{\prime}\cap\cS_{2}^{\prime}.
	\end{aligned}
	\end{equation*}

	Recall the definition of $\cS$ in \eqref{def: proper set}. According to Lemmas \ref{lem: close set} and \ref{lem: sure screening}, we have  $\|\hat{\theta} - \hat{\theta}_{\rm IVW}\|\leq \epsilon_{n}$ and $\hat{b}_{k} = 0$ for $k \in \cK_{0}$ when $\cS$ holds.It is straightforward to verify that the conditions of Lemma \ref{lem: main} is satisfied under the conditions of this theorem. Hence we have $P(\cS) \to 1$ according to Lemma \ref{lem: main}. To prove this theorem, it then suffices to prove $\hat{b}_{k}\not = 0$ for $k\in \cK_{0}^{c}$ on $\cS^{\prime}\cap \cS$ and $P(\cS^{\prime}) \to 1$. 
	
	First, we prove that $\hat{b}_{k}\not = 0$ for $k\in \cK_{0}^{c}$ on the event $\cS^{\prime}\cap \cS$. The arguments in the rest of this paragraph are derived on the event $\cS^{\prime}\cap \cS$.
	By Weyl's Theorem, $\max_{k}\{\max\{|\lambda_{\rm min}(\tilde{V}_{k}) - \lambda_{\rm min}(V_{k}^{*})|, |\lambda_{\rm max}(\tilde{V}_{k}) - \lambda_{\rm max}(V_{k}^{*})|\}\} \leq \max_{k}\|\tilde{V}_{k} - V_{k}^{*}\| \leq \Delta_{\rm M}$.
	Thus, for $k = 1,\dots,K$, 
	\[C_{\rm L} \leq \lambda_{\rm min}(\tilde{V}_{k}) \leq \lambda_{\rm max}(\tilde{V}_{k}) \leq C_{\rm U}.\]
	Because $(\hat{\theta}^{\T},\hat{b}_{1}^{\T},\dots,\hat{b}_{K}^{\T})^{\T}$ is a minimum point of problem \eqref{eq: opt problem} in the main text, we have \begin{equation}\label{eq: KKT}
	\tilde{\pi}_{k}\tilde{V}_{k}(\tilde{\theta}_{k} - \hat{\theta} - \hat{b}_{k}) = \lambda\tilde{w}_{k}\hat{z}_{k}, \ k \in \cK_{0}^{c}
	\end{equation}
	for some $\|\hat{z}_{k}\| \leq  1$ according to the KKT condition. Thus by \eqref{eq: KKT}, the definition of eigenvalue and the triangular inequality, we have
	\[\tilde{\pi}_{k}C_{\rm L}(|\|\tilde{\theta}_{k} - \hat{\theta}\| - \|\hat{b}_{k}\||)\leq
	\tilde{\pi}_{k}\lambda_{\rm min}(\tilde{V}_{k})(|\|\tilde{\theta}_{k} - \hat{\theta}\| - \|\hat{b}_{k}\||)
	\leq \|\tilde{\pi}_{k}\tilde{V}_{k}(\tilde{\theta}_{k} - \hat{\theta} - \hat{b}_{k})\| \leq \lambda\tilde{w}_{k}\]
	for $k\in \cK_{0}^{c}$. Then on $\cS_{2}^{\prime}$
	\begin{align*}
	\|\hat{b}_{k}\| 
	&\geq \|\tilde{\theta}_{k} - \hat{\theta}\| - |\|\tilde{\theta}_{k} - \hat{\theta}\| - \|\hat{b}_{k}\||\\
	& \geq \|\tilde{\theta}_{k} - \hat{\theta}\| - (\tilde{\pi}_{k}C_{\rm L})^{-1}\lambda\tilde{w}_{k}\\
	& \geq \|\tilde{\theta}_{k} - \hat{\theta}_{\rm IVW}\| - \|\hat{\theta} - \hat{\theta}_{\rm IVW}\| - (\tilde{\pi}_{k}C_{\rm L})^{-1}\lambda\tilde{w}_{k}\\
	&\geq \|\tilde{\theta}_{k} - \hat{\theta}_{\rm IVW}\| - \epsilon_{n} - (\tilde{\pi}_{k}C_{\rm L})^{-1}\lambda\tilde{w}_{k}\\
	&> 0.
	\end{align*}
	This indicates that $\hat{b}_{k}\not = 0$ for $k\in \cK_{0}^{c}$ on the event $\cS^{\prime}\cap \cS$.
	
	Next, we prove that $P(\cS^{\prime}) \to 1$. By Condition \ref{cond: matrix rate 2}, we have $P(\cS_{1}^{\prime}) \to 1$.	
	Note that
	\[\hat{\theta}_{\rm IVW} - \tilde{\theta}_{k} = \left(\sum_{j \in \cK_{0}}\tilde{\pi}_{j}\tilde{V}_{j}\right)^{-1}\left(\sum_{j \in \cK_{0}}\tilde{\pi}_{j}\tilde{V}_{j}(\tilde{\theta}_{j} - \tilde{\theta}_{k})\right).\]
	On the event $\cS_{1}^{\prime}$, for $k \in \cK_{0}^{c}$, we have
	\begin{align*}
	\|\hat{\theta}_{\rm IVW} - \tilde{\theta}_{k}\|& \geq (B_{\rm U} + \max_{j\in \cK_{0}}\|\tilde{V}_{j} - V_{j}^{*}\|)^{-1}B_{\rm L}(\min_{j\in \cK_{0}}\|V_{j}^{*}(\tilde{\theta}_{j} - \tilde{\theta}_{k})\|)\\
	& \geq (B_{\rm U} + \max_{j\in cK_{0}}\|\tilde{V}_{j} - V_{j}^{*}\|)^{-1}B_{\rm L}(B_{\rm L} - \max_{j\in \cK_{0}}\|\tilde{V}_{j} - V_{j}^{*}\|)\min_{j\in \cK_{0}}\|\tilde{\theta}_{j} - \tilde{\theta}_{k}\|\\
	&\geq C_{\rm U}^{-1}B_{\rm L}C_{\rm L}\min_{j\in \cK_{0}}\|\tilde{\theta}_{j} - \tilde{\theta}_{k}\|.
	\end{align*}
	Because for $k \in \cK_{0}^{c}$ and $j \in \cK_{0}$,
	\[\|\tilde{\theta}_{j} - \tilde{\theta}_{k}\| \geq \|b_{k}^{*}\| - \|\tilde{\theta}_{j} - \theta_{0}\| - \|\tilde{\theta}_{k} - \theta_{k}^{*}\|,\]
	we have
	\[
	\min_{j\in \cK_{0},k\in \cK_{0}^{c}}\|\tilde{\theta}_{j} - \tilde{\theta}_{k}\| \geq B_{\rm L} - 2 \max_{j}\|\tilde{\theta}_{j} - \theta_{j}^{*}\|.
	\]
	Thus
	\begin{align*}
	\cS_{1}^{\prime} \cap \cS_{2}^{\prime} 
	&= \cS_{1}^{\prime} \cap \{\|\hat{\theta}_{\rm IVW} - \tilde{\theta}_{k}\| > (\tilde{\pi}_{k}C_{\rm L})^{-1}\lambda\tilde{w}_{k} + \epsilon_{n}, \ k\in \cK_{0}^{c}\}\\
	&\supset \cS_{1}^{\prime} \cap \{C_{\rm U}^{-1}B_{\rm L}C_{\rm L}(B_{\rm L} - 2 \max_{j}\|\tilde{\theta}_{j} - \theta_{j}^{*}\|) > (\min_{k\in\cK_{0}^{c}}\tilde{\pi}_{k}C_{\rm L})^{-1}\lambda\max_{k\in \cK_{0}^{c}}\tilde{w}_{k} + \epsilon_{n}\}\\
	&\equalscolon \cS_{1}^{\prime}\cap\cS_{2}^{\prime*}.
	\end{align*}
	According to Condition \ref{cond: estimate rate}, $C_{\rm U}^{-1}B_{\rm L}C_{\rm L}(B_{\rm L}
	- 2 \max_{j}\|\tilde{\theta}_{j} - \theta_{j}^{*}\|) = C_{\rm U}^{-1}C_{\rm L}B_{\rm L}^{2} + o_{P}(1)$.
	Recall that $\tilde{w}_{k} = 1 / \|\tilde{b}_{k}\|^{\alpha}$. Then $\max_{k\in \cK_{0}^{c}}\tilde{w}_{k} = O_{P}(1)$ according to \eqref{eq: bias est} and \eqref{eq: estimate rate}. Since $\min_{k\in \cK_{0}^{c}}\tilde{\pi}_{k} > C_{\pi} /K$ and $\lambda \asymp 1/n$, we have $(\min_{k\in\cK_{0}^{c}}\tilde{\pi}_{k}C_{\rm L})^{-1}\lambda\max_{k\in \cK_{0}^{c}}\tilde{w}_{k} = O_{P}(K/n)$. Moreover, $\epsilon_{n} = o(K\log n / n)$ by definition. $(\min_{k\in\cK_{0}^{c}}\tilde{\pi}_{k}C_{\rm L})^{-1}\lambda\max_{k\in \cK_{0}^{c}}\tilde{w}_{k} + \epsilon_{n} = o_{P}(1)$ because $K\log n / n \to 0$. Thus $P(\cS_{2}^{\prime *}) \to 1$ and hence $P(\cS^{\prime}) \to 1$. This completes the proof.
	
\end{proof}

\section{Proof of Theorem \ref{thm: AN}}
\paragraph{\bf Restate of Theorem \ref{thm: AN}.}
\emph{Suppose Conditions \ref{cond: existence of nonbias}, \ref{cond: estimate rate} and \ref{cond: uniform AL} hold. If (i) $\nu_{1} < 1/2$; (ii) there are some deterministic matrices $V_{k}^{*}$, $k\in \cK_{0}$, such that $\max_{k\in \cK_{0}}\|\tilde{V}_{k} - V_{k}^{*}\| = o_{P}(n^{- 1 / 2 + \nu_{2}})$; (iii) for $k\in \cK_{0}$, the eigenvalues of $V_{k}^{*}$ and  $\var\left[\Psi_{k}(Z^{(k)})\right]$ are bounded away from zero and infinity; (iv) for $k\in \cK_{0}$, $u\in \bbR^{d}$, $\|u\| = 1$ and some $\tau > 0$, $E[|u^{\T}\Psi_{k}(Z^{(k)})|^{1 + \tau}]$ are bounded; (v)$\lambda \asymp 1/n$ and $\alpha > \max\{\nu_{1}\nu_{2}^{-1}, \nu_{2}^{-1} - 1\}$, then for any fixed $q$ and $q\times d$ matrix $W_{n}$ such that the eigenvalues of $W_{n}W_{n}^{\T}$ are bounded away from zero and infinity, we have
	\[\sqrt{n}\cI_{n}^{-1/2}W_{n}(\hat{\theta} - \theta_{0})\stackrel{d}{\to} N(0, I_{q}),\]
	where $I_{q}$ is the identity matrix of order $q$, $\cI_{n} = \sum_{k \in \cK_{0}}\tilde{\pi}_{k}H_{n,k}\var\left[ \Psi_{k}(Z^{(k)})\right]H_{n,k}^{\T}$, $H_{n,k} = W_{n}V_{0}^{*-1}V_{k}^{*}$ and $V_{0}^{*} = \sum_{k\in \cK_{0}} \tilde{\pi}_{k}V_{k}^{*}$.}
\begin{proof}
	According to Theorem \ref{thm: equivalence}, under (i), (ii), (v), Condition \ref{cond: existence of nonbias} and \ref{cond: estimate rate}, we have
	\[\|\hat{\theta} - \hat{\theta}_{\rm IVW}\| = o_{P}\left(\frac{1}{\sqrt{n}}\right).\]
	Then to establish the asymptotic normality result of $\hat{\theta}$, it suffices to establish the asymptotic normality of $\hat{\theta}_{\rm IVW}$.
	According to (ii), (iii) and Condition \ref{cond: existence of nonbias}, we have
	\begin{align*}
	&\cI_{n}^{-1/2}W_{n}(\hat{\theta}_{\rm IVW} - \theta_{0})\\
	&=  \cI_{n}^{-1/2}W_{n}\left(\sum_{k \in \cK_{0}}\tilde{\pi}_{k}\tilde{V}_{k}\right)^{-1}\left(\sum_{k \in \cK_{0}}\tilde{\pi}_{k}\tilde{V}_{k}(\tilde{\theta}_{k} - \theta_{0})\right)\\
	&=  \cI_{n}^{-1/2}W_{n}\left(V_{0}^{*-1} + o_{P}(1)\right)\left(\sum_{k \in \cK_{0}}\tilde{\pi}_{k}V_{k}^{*}(\tilde{\theta}_{k} - \theta_{0})\right)\\
	&\quad+ \cI_{n}^{-1/2}W_{n}\left(V_{0}^{*-1} + o_{P}(1)\right)\left(\sum_{k \in \cK_{0}}\tilde{\pi}_{k}\left(\tilde{V}_{k} - V_{k}^{*}\right)(\tilde{\theta}_{k} - \theta_{0})\right),
	\end{align*}
	where $V_{0}^{*} = \sum_{k\in \cK_{0}} \tilde{\pi}_{k}V_{k}^{*}$. According to (iii) and (iv), we have
	\begin{equation}\label{eq: rate variance}
	\cI_{n}^{-1/2}W_{n} = O(1).
	\end{equation} 
	Because $\max_{k}\|\tilde{\theta}_{k} - \theta_{k}^{*}\| = O_{P}\left(n^{-\nu_{2}}\right)$, $\max_{k\in \cK_{0}}\|\tilde{V}_{k} - V_{k}^{*}\| = o_{P}(n^{- 1 / 2 + \nu_{2}})$ and \eqref{eq: rate variance}, we have
	\begin{align*}
	&\cI_{n}^{-1/2}W_{n}\left(V_{0}^{*-1} + o_{P}(1)\right)\left(\sum_{k \in \cK_{0}}\tilde{\pi}_{k}\left(\tilde{V}_{k} - V_{k}^{*}\right)(\tilde{\theta}_{k} - \theta_{0})\right) \\
	&= O_{P}\left(\max_{k}\|\tilde{\theta}_{k} - \theta_{k}^{*}\|\right)O_{P}\left(\max_{k\in \cK_{0}}\|\tilde{V}_{k} - V_{k}^{*}\|\right) = o_{P}\left(\frac{1}{\sqrt{n}}\right),
	\end{align*}
	and hence
	\begin{align*}
	&\cI_{n}^{-1/2}W_{n}(\hat{\theta}_{\rm IVW} - \theta_{0})\\
	&=\cI_{n}^{-1/2}W_{n}V_{0}^{*-1}\left(\sum_{k \in \cK_{0}}\tilde{\pi}_{k}V_{k}^{*}(\tilde{\theta}_{k} - \theta_{0})\right)+ o_{P}\left(\cI_{n}^{-1/2}W_{n}V_{0}^{*-1}\left(\sum_{k \in \cK_{0}}\tilde{\pi}_{k}V_{k}^{*}(\tilde{\theta}_{k} - \theta_{0})\right)\right)\\
	&\quad+ o_{P}\left(\frac{1}{\sqrt{n}}\right).
	\end{align*}
	Thus Theorem \ref{thm: AN} is proved if we show
	\begin{equation}\label{eq: main term AN}
	\sqrt{n}\cI_{n}^{-1/2}W_{n}V_{0}^{*-1}\left(\sum_{k \in \cK_{0}}\tilde{\pi}_{k}V_{k}^{*}(\tilde{\theta}_{k} - \theta_{0})\right)
	\stackrel{d}{\to} N(0, I_{q}). 
	\end{equation}
	By Condition \ref{cond: uniform AL}, we have
	\begin{align*}
	&\sqrt{n}\cI_{n}^{-1/2}W_{n}V_{0}^{*-1}\left(\sum_{k \in \cK_{0}}\tilde{\pi}_{k}V_{k}^{*}(\tilde{\theta}_{k} - \theta_{0})\right)\\
	&=\sqrt{n}\cI_{n}^{-1/2}W_{n}V_{0}^{*-1}\left\{\sum_{k \in \cK_{0}}\tilde{\pi}_{k}V_{k}^{*}\frac{1}{n_{k}}\sum_{i=1}^{n_{k}}\Psi_{k}(Z_{i}^{(k)})\right\} + o_{P}(1) \\
	& = \sum_{k\in\cK_{0}}\sum_{i=1}^{n_{k}}\eta_{k,i} + o_{P}(1),
	\end{align*}
	where $\eta_{k,i} =  \cI_{n}^{-1/2}W_{n}V_{0}^{*-1}V_{k}^{*}\Psi_{k}(Z_{i}^{(k)})/\sqrt{n}$. Because $E\left[\Psi_{k}(Z^{(k)})\right] = 0$ for $k \in \cK_{0}$,  (iii) and (iv) implies that Lindeberg-Feller condition \citep{van2000asymptotic} is satisfied. Then \eqref{eq: main term AN} follows since 
	\[
	\sum_{k\in \cK_{0}}\sum_{i=1}^{n_{k}}\var[\eta_{k,i}] = \cI_{n}^{-1/2}\cI_{n}\cI_{n}^{-1/2} = I_{q}.
	\]
	
\end{proof}

\section{Uniform asymptotically linear representation of M-estimator}\label{app: M-estimator}
In this section, we establish the uniform asymptotically linear representation in the case where $\tilde{\theta}_{k}$'s are M-estimators, i.e.
\[\tilde{\theta}_{k} = \mathop{\arg\min}_{\theta}\frac{1}{n_{k}}\sum_{i=1}^{n_{k}}L_{k}(Z_{i}^{(k)}, \theta),\]
for $k=1,\dots,K$, where $L_{k}(\cdot,\cdot)$ is some loss function that may differ from source to source.
In this case, the probability limit of $\tilde{\theta}_{k}$ is
the minimum point of $E[L_{k}(Z^{(k)}; \theta)]$ under some regularity conditions. Hence here we use $\theta_{k}^{*}$ to denote the minimum point of $E[L_{k}(Z^{(k)}; \theta)]$. Let $\zeta_{k}(\theta) = L_{k}(Z^{(k)}; \theta) - E[L_{k}(Z^{(k)}; \theta)]$. 
To begin with, we first introduce a commonly used condition in the literature with diverging parameter dimension.
\begin{condition}\label{cond: loss gradient and global}
	There are some constants $\sigma_{1}$, $u_{1}$, $b$ independent of $n$ and some positive definite matrices $\Phi_{k}$ ($k=1,\dots,K$) that may depend on $n$ such that
	\begin{equation*}
	E\left[\exp\left(\lambda\frac{\gamma^{\T}\nabla\zeta_{k}(\theta)}{\|\Phi_{k}\gamma\|}\right)\right] \leq \exp\left(\frac{\sigma_{1}^{2}\lambda^{2}}{2}\right)
	\end{equation*}
	and 
	\begin{equation*}
	E[L_{k}(Z^{(k)},\theta)] - E[L_{k}(Z^{(k)},\theta_{k}^{*})] \geq b\|\Phi_{k}(\theta - \theta_{k}^{*})\|^{2},
	\end{equation*}
	for all $\theta$, $|\lambda| \leq u_{1}$, $\|\gamma\| = 1$ and $k=1,\dots,K$.
\end{condition}
See \cite{spokoiny2012parametric,spokoiny2013bernstein,zhou2018new, chen2020robust} for further explanations and examples of this condition. The following proposition shows that Condition \ref{cond: loss gradient and global} along with some other conditions implies Condition \ref{cond: estimate rate}.


\begin{proposition}\label{prop: uniform consistent}
	Under Condition \ref{cond: loss gradient and global}, if (i) for $k = 1,\dots, K$, the eigenvalues of $\Phi_{k}$ are bounded away from zero; (ii) $d = O(n^{\nu_{0}})$, $K = O(n^{\nu_{1}})$ for some positive constants $\nu_{0}, \nu_{1}$ such that $\nu_{0} + \nu_{1} < 1$ and (iii) $\tilde{\pi}_{k} \geq C^{*}/K$ for some positive constant $C^{*}$, then $\max_{k}\|\tilde{\theta}_{k} - \theta_{k}^{*}\| = O_{P}\left(n^{-(1 - \nu_{0} - \nu_{1})/2}\right)$.
\end{proposition}
\begin{proof}
	For convenience, in this and the following proofs, we let $C$ be a generic positive constant that
	may be different in different places.
	Under Condition \ref{cond: loss gradient and global}, according to Theorem 5.2 in \cite{spokoiny2012parametric}, we have 
	\[P\left(\|\Phi_{k}(\tilde{\theta}_{k} - \theta_{k}^{*})\| \geq 6\sigma_{1} b^{-1}\sqrt{\frac{3d + t}{n_{k}}}\right) \leq e^{-t},\]
	for $t\leq \Delta_{k}$ with $\Delta_{k} = (3b^{-1}\sigma_{1}^{2}u_{1}n_{k}^{1/2}-1)^{2} - 3d$.
	Thus, according to (i), we have
	\[P\left(\|\tilde{\theta}_{k} - \theta_{k}^{*}\| \geq L\sqrt{\frac{3d + t}{n_{k}}}\right) \leq e^{-t},\]
	for $t\leq \Delta_{k}$  and $k = 1,\dots, K$, where $L = 6\sigma_{1} b^{-1}B_{L}^{-1}$. By Bonferroni inequality, it follows
	\[P\left(\max_{k}\|\tilde{\theta}_{k} - \theta_{k}^{*}\|\geq L\max_{k}\left\{\sqrt{\frac{3d + t}{n_{k}}}\right\}\right) \leq Ke^{-t}.\]
	By (iii) we have $n_{k}\geq C^{*} n / K$.
	Letting $t_{n} = \min\left\{n^{\nu_{0}},\min_{k}\Delta_{k}\right\}$,  according to (ii) and (iii), we have \[\max_{k}\left\{\sqrt{(3d + t_{n})/n_{k}}\right\} \leq \sqrt{K(3d + t_{n})/(C^{*}n)} \leq C\left(n^{-\frac{1 - \nu_{0}- \nu_{1}}{2}}\right).\]
	Thus
	\begin{equation}\label{eq: max rate-diverge K}
	P\left(\max_{k}\|\tilde{\theta}_{k} - \theta_{k}^{*}\|\geq Cn^{-\frac{1 - \nu_{0} - \nu_{1}}{2}}\right) \leq K\exp\left(-t_{n}\right) \to 0.
	\end{equation}
	This indicates that $\max_{k}\|\tilde{\theta}_{k} - \theta_{k}^{*}\| = O_{P}\left(n^{-(1 - \nu_{0} - \nu_{1})/2}\right)$.
	
\end{proof}
To establish the uniform asymptotically linear representation, some further conditions on the Hessian of the expected loss function are required. Let $D_{k}(\theta) = (\nabla^{2}E[L_{k}(Z^{(k)},\theta)])^{1/2}$ be the Hessian of the expected loss function and let $D_{k*} = D_{k}(\theta_{k}^{*})$. 
\begin{condition}\label{cond: loss hessian}
	For $k \in \cK_{0}$, the eigenvalues of $D_{k*}$ are bounded away from zero and infinity, and there is some constant $M_{*}$ such that
	$\|D_{k}^{2}(\theta) - D_{k*}^{2}\| \leq M_{*}\|\theta - \theta_{0}\|$ for all $\theta$. Moreover, for some constants $\sigma_{2}$ and $u_{2}$, 
	\[E\left[\exp\left(\lambda\frac{\gamma_{1}^{\T}\nabla^{2}\zeta_{k}(\theta)\gamma_{2}}{\|D_{k*}\gamma_{1}\|\|D_{k*}\gamma_{2}\|}\right)\right] \leq \exp\left(\frac{\sigma_{2}^{2}\lambda^{2}}{2}\right)\]
	for all $|\lambda| \leq u_{2}$, $\|\gamma_{1}\| = 1$, $\|\gamma_{2}\| = 1$ and $k\in \cK_{0}$.
\end{condition}
Under Condition \ref{cond: loss hessian} and the conditions of Proposition \ref{prop: uniform consistent}, we establish the uniform asymptotically linear representation (Condition \ref{cond: uniform AL}).
\begin{proposition}\label{prop: uniform i.i.d.}
	Under Condition \ref{cond: loss hessian} and the conditions of Proposition \ref{prop: uniform consistent}, if $\nu_{0} + \nu_{1} < 1/2$, then Condition \ref{cond: uniform AL} holds with $\Psi_{k}(Z^{(k)}) = - D_{k*}^{-2}\nabla L(Z^{(k)}, \theta_{0})$
\end{proposition}
\begin{proof}
	For $k\in \cK_{0}$, according to Condition \ref{cond: loss hessian}, it is not hard to verify that the Condition $ED_{2}$ in \citep{spokoiny2013bernstein} is satisfied with $g=u_{2}\sqrt{n_{k}}$ and $\omega = 1/\sqrt{n_{k}}$. Because $B_{\rm L}\leq \min_{k}\lambda_{\rm min}(D_{k*}) \leq \max_{k}\lambda_{\rm max}(D_{k*}) \leq B_{\rm U}$ and 	$\|D_{k}^{2}(\theta) - D_{k*}^{2}\| \leq M_{*}\|\theta - \theta_{0}\|$, we have
	\[
	\begin{aligned}
	\|D_{k*}^{-1}D_{k}^{2}(\theta)D_{k*}^{-1} - I_{d}\|& \leq \|D_{k*}^{-1}\|^{2} \|D_{k}^{2}(\theta) - D_{k*}^{2}\|
	&\leq M_{*} B_{\rm L}^{-3} \|D_{k*}(\theta - \theta_{0})\|,
	\end{aligned}
	\]
	for $k \in \cK_{0}$, where $I_{d}$ is the $d\times d$ identity matrix.
	
	Thus Condition $\cL$ in \cite{spokoiny2013bernstein} is satisfied with $\delta(r) = M_{*}B_{\rm L}^{-3}r$. For $k \in \cK_{0}$, define the event
	\[
	E_{r,t}^{(k)} = \left\{\sup_{\theta \in \Theta_{*}(r)} \left\|\frac{1}{n_{k}}\sum_{i=1}^{n}D_{k*}^{-1}\left\{\nabla L(Z_{i}^{(k)}, \theta) - \nabla L(Z_{i}^{(k)}, \theta_{0})\right\} - D_{k*}(\theta - \theta_{0})\right\|
	\geq \epsilon_{r, t}^{(k)}\right\},
	\]
	where $\Theta_{*}(r) = \|D_{k*}(\theta - \theta_{0})\| \leq r$ and $\epsilon_{r, t}^{(k)} = M_{*}B_{\rm U}^{3}r^{2} + 6 \sigma_{2}r \sqrt{(4p + 2t) / n_{k}}$.
	According to Proposition 3.1 in \cite{spokoiny2013bernstein}, we have 
	\begin{equation}\label{eq: single expansion bound}
	P\left(E_{r,t}^{(k)}\right)\leq \exp(-t)
	\end{equation}
	for $k\in \cK_{0}$ and $t \leq \Delta_{k}^{\prime}$ with $\Delta_{k}^{\prime} = -2p + u_{2}n_{k} / 2$. By \eqref{eq: max rate-diverge K}, there is some $C$ such that 
	\[P\left(\max_{k}\|\tilde{\theta}_{k} - \theta_{k}^{*}\|\geq Cn^{-(1 - \nu_{0} - \nu_{1})/2}\right) \to 0.\]
	Let $r_{n} = B_{\rm U}Cn^{-(1 - \nu_{0} - \nu_{1})/2}$. Then we have
	\[\bigcup_{k\in \cK_{0}}\left\{\tilde{\theta}_{k} \notin \Theta(r_{n})\right\} \subset \left\{\max_{k}\{\|\tilde{\theta}_{k} - \theta_{k}^{*}\|\}\geq Cn^{-\frac{1 - \nu_{0} - \nu_{1}}{2}}\right\}.\]
	This implies 
	\begin{equation}\label{eq: prob localization}
	P\left(\bigcup_{k\in \cK_{0}}\left\{\tilde{\theta}_{k} \notin \Theta(r_{n})\right\}\right)\leq P\left(\max_{k}\{\|\tilde{\theta}_{k} - \theta_{k}^{*}\|\}\geq Cn^{-\frac{1 - \nu_{0} - \nu_{1}}{2}}\right)
	\to 0.
	\end{equation}
	Letting $t_{n} = \min\left\{n^{\nu_{0}},\min_{k}\Delta_{k}^{\prime}\right\}$, we have 
	\begin{equation}\label{eq: union expansion bound}
	P\left(\bigcup_{k\in \cK_{0}}E_{r_{n}, t_{n}}^{(k)}\right) \leq K\exp(-t_{n}) \to 0
	\end{equation}
	according to \eqref{eq: single expansion bound} and the rate conditions on $K$.
	By the definition of $\tilde{\theta}_{k}$, we have $\sum_{i=1}^{n_{k}}\nabla L(Z_{i}^{(k)}, \tilde{\theta}_{k})/n_{k} = 0$. Combining this with \eqref{eq: prob localization} and \eqref{eq: union expansion bound}, we have
	\begin{equation}\label{eq: asymptotically linear}
	P\left(\max_{k\in \cK_{0}}\left\|\frac{1}{n_{k}}\sum_{i=1}^{n}D_{k*}^{-1}\nabla L(Z_{i}^{(k)}, \theta_{0}) + D_{k*}(\tilde{\theta}_{k} - \theta_{0})\right\| \geq \xi_{n}\right) \to 0
	\end{equation}
	where $\xi_{n} = M_{*}B_{\rm U}^{3}r_{n}^{2} + 6 \sigma_{2}r_{n} \sqrt{(4d + 2t_{n}) / n_{k}} = o(1/\sqrt{n})$ because $\nu_{0} + \nu_{1} < 1/2$. Thus
	\begin{equation}\label{eq: uniform asymptotically linear}
	\max_{k\in \cK_{0}}\left\|\frac{1}{n_{k}}\sum_{i=1}^{n}D_{k*}^{-2}\nabla L(Z_{i}^{(k)}, \theta_{0}) + (\tilde{\theta}_{k} - \theta_{0})\right\| = o_{P}\left(\frac{1}{\sqrt{n}}\right),
	\end{equation}
	and this implies the result of the proposition.
	
\end{proof}

	\bibliographystyle{chicago}
	\bibliography{rm}
	\newpage
    \begin{table}
    	\caption{\label{table: sim1}NB and SSE with least squares regression in the presence of biased sources (results are multiplied by 10)}
    	\centering
    	\begin{tabular}{*{12}{c}}
    		\toprule
    		\multirow{2}{*}{Estimator} & \multirow{2}{*}{$n_{*}$}  & \multicolumn{2}{c}{naive} & \multicolumn{2}{c}{oracle} &\multicolumn{2}{c}{\emph{iFusion}}& \multicolumn{2}{c}{$\tilde{\theta}$} & \multicolumn{2}{c}{$\hat{\theta}$} \\
    		& &\small NB &\small SSE &\small NB &\small SSE &\small NB 
    		&\small SSE&\small NB &\small SSE &\small NB &\small SSE\\
    		\midrule
    		\multirow{3}{*}{$d=3, K=10$}& 100 & 9.87 & 0.55 & 0.10 & 1.31 &0.09 & 1.34 & 0.46 & 1.40 & 0.10 & 1.31\\
    		& 200 & 9.85 & 0.40 & 0.04 & 0.90 & 0.04 & 0.92 & 0.37 & 0.98 & 0.04 & 0.90\\
    		& 500 & 9.85 & 0.24 & 0.02 & 0.54 & 0.02 & 0.54 & 0.21 & 0.62 & 0.02 & 0.54\\
    		\hline
    		\multirow{3}{*}{$d=3, K=30$} & 100 & 9.85 & 0.33 & 0.05 & 0.76 & 0.06 & 0.87 & 0.56 & 0.75 & 0.03 & 0.76\\
    		& 200 & 9.85 & 0.21 & 0.01 & 0.50 & 0.01 & 0.57 & 0.42 & 0.49 & 0.02 & 0.50\\
    		& 500 & 9.84 & 0.14 & 0.01 & 0.32 & 0.01 & 0.33 & 0.26 & 0.35 & 0.01 & 0.32\\
    		\hline
    		\multirow{3}{*}{$d=18, K=10$} & 100 & 24.13 & 3.64 & 0.13 & 8.09 & 0.16 & 11.23 & 1.18 & 7.26 & 0.14 & 8.07 \\
    		& 200 & 24.12 & 2.41 & 0.10 & 5.47 & 0.14 & 7.29 & 0.77 & 5.23 & 0.08 & 5.46 \\
    		& 500 & 24.13 & 1.50 & 0.05 & 3.29 & 0.06 & 3.81 & 0.50 & 3.27 & 0.07 & 3.29\\
    		\hline
    		\multirow{3}{*}{$d=18, K=30$} & 100 & 24.13 & 2.14 & 0.07 & 4.78 & 0.16 & 11.37 & 1.44 & 4.00 & 0.11 & 4.75\\
    		& 200 & 24.11 & 1.38 & 0.05 & 3.16 & 0.15 & 7.71 & 1.00 & 2.81 & 0.06 & 3.15\\
    		& 500 & 24.13 & 0.87 & 0.03 & 1.89 & 0.07 & 4.51 & 0.66 & 1.76 & 0.04 & 1.89\\
    		\bottomrule
    	\end{tabular}
    \end{table}
~\\
    \newpage
\begin{table}
	\caption{\label{table: sim2}NB and SSE with least squares regression and no biased sources (results are multiplied by 10)}
	\centering
	\begin{tabular}{*{10}{c}}
		\toprule
		\multirow{2}{*}{Estimator} & \multirow{2}{*}{$n_{*}$}  & \multicolumn{2}{c}{oracle} &\multicolumn{2}{c}{\emph{iFusion}}& \multicolumn{2}{c}{$\tilde{\theta}$} & \multicolumn{2}{c}{$\hat{\theta}$} \\
		& &\small NB &\small SSE &\small NB &\small SSE &\small NB &\small SSE &\small NB &\small SSE\\
		\midrule 
		\multirow{3}{*}{$d=3, K=10$}& 100 & 0.03 & 0.55 & 0.03 & 0.56 & 0.03 & 0.60 & 0.03 & 0.55\\
		& 200 & 0.01 & 0.40 & 0.01 & 0.40 & 0.01 & 0.44 & 0.01 & 0.40\\
		& 500 & 0.00 & 0.24 & 0.00 & 0.24 & 0.00 & 0.26 & 0.00 & 0.24\\
		\hline
		\multirow{3}{*}{$d=3, K=30$} & 100 & 0.01 & 0.33 & 0.02 & 0.46 & 0.01 & 0.35 & 0.01 & 0.33\\
		& 200 & 0.01 & 0.21 & 0.01 & 0.29 & 0.01 & 0.23 & 0.01 & 0.21\\
		& 500 & 0.01 & 0.14 & 0.00 & 0.15 & 0.01 & 0.15 & 0.01 & 0.14\\
		\hline
		\multirow{3}{*}{$d=18, K=10$} & 100 & 0.05 & 3.64 & 0.16 & 10.13 & 0.05 & 3.66 & 0.05 & 3.64\\
		& 200 & 0.04 & 2.41 & 0.09 & 5.05 & 0.04 & 2.44 & 0.04 & 2.41\\
		& 500 & 0.02 & 1.50 & 0.03 & 1.91 &  0.03 & 1.52 & 0.02 & 1.50\\
		\hline
		\multirow{3}{*}{$d=18, K=30$} & 100 & 0.03 & 2.14 & 0.16 & 11.33 & 0.03 & 2.15 & 0.03 & 2.14\\
		& 200 & 0.03 & 1.38 & 0.14 & 7.49 & 0.03 & 1.39 & 0.03 & 1.38\\
		& 500 & 0.02 & 0.87 & 0.05 & 3.75 & 0.02 & 0.88 & 0.02 & 0.87\\
		\bottomrule
	\end{tabular}
\end{table}
~\\
	\newpage
\begin{table}
	\caption{\label{table: sim3}NB and SSE with logistic regression in the presence of biased sources (results are multiplied by 10)}
	\centering
	\begin{tabular}{*{12}{c}}
		\toprule
		\multirow{2}{*}{Estimator} & \multirow{2}{*}{$n_{*}$}  & \multicolumn{2}{c}{naive} & \multicolumn{2}{c}{oracle} &\multicolumn{2}{c}{\emph{iFusion}} &\multicolumn{2}{c}{$\tilde{\theta}$} & \multicolumn{2}{c}{$\hat{\theta}$} \\
		& &\small NB &\small SSE &\small NB &\small SSE &\small NB &\small SSE
		&\small NB &\small SSE &\small NB &\small SSE\\
		\midrule
		\multirow{3}{*}{$d=3, K=10$}& 100 & 4.79 & 2.32 & 0.24 & 2.73 & 0.24 & 3.80 & 0.84 & 2.48 & 0.15 & 2.78\\
		& 200 & 5.89 & 1.89 & 0.13 & 1.83 & 0.18 & 2.34 & 0.58 & 1.88 & 0.09 & 1.89\\
		& 500 & 7.36 & 1.62 & 0.02 & 1.10 & 0.05 & 1.35  & 0.42 & 1.23 & 0.02 & 1.12\\
		\hline
		\multirow{3}{*}{$d=3, K=30$} & 100 & 4.77 & 1.34 & 0.22 & 1.62 & 0.28 & 3.86 & 0.99 & 1.37 & 0.13 & 1.61\\
		& 200 & 5.88 & 1.11 & 0.10 & 1.08 & 0.16 & 2.21 & 0.73 & 0.97 & 0.09 & 1.08\\
		& 500 & 7.37 & 0.84 & 0.06 & 0.65 & 0.03 & 1.20 & 0.49 & 0.67 & 0.06 & 0.65\\
		\hline
		\multirow{3}{*}{$d=18, K=10$} & 100 & 6.19 & 16.47 & 1.97 & 24.68 & 1.90 & 35.18 & 2.54 & 16.44 & 1.48 & 24.57\\
		& 200 & 7.12 & 11.91 & 0.82 & 13.70 & 0.93 & 19.59 & 1.75 & 11.04 & 0.80 & 13.64\\
		& 500 & 9.11 & 8.13 & 0.28 & 8.04 & 0.31 & 11.34 & 1.07 & 7.13 & 0.27 & 8.01\\
		\hline
		\multirow{3}{*}{$d=18, K=30$} & 100 & 6.24 & 9.63 & 1.86 & 14.24 & 1.90 & 35.18 & 2.85 & 9.25 & 1.43 & 13.54\\
		& 200 & 7.15 & 6.96 & 0.78 & 8.03 &
		0.93 & 19.59 & 2.04 & 6.16 & 0.74 & 7.98\\
		& 500 & 9.12 & 4.62 & 0.29 & 4.59 & 0.31 & 11.34 & 1.33 & 3.91 & 0.27 & 4.57\\
		\bottomrule
	\end{tabular}
\end{table}

	\newpage
\begin{table}
	\caption{\label{table: sim4}NB and SSE with logistic regression and no biased sources (results are multiplied by 10)}
	\centering
	\begin{tabular}{*{10}{c}}
		\toprule
		\multirow{2}{*}{Estimator} & \multirow{2}{*}{$n_{*}$}  & \multicolumn{2}{c}{oracle} & \multicolumn{2}{c}{\emph{iFusion}} & \multicolumn{2}{c}{$\tilde{\theta}$} & \multicolumn{2}{c}{$\hat{\theta}$} \\
		& &\small NB &\small SSE &\small NB &\small SSE &\small NB &\small SSE &\small NB &\small SSE\\
		\midrule
		\multirow{3}{*}{$d=3, K=10$}& 100 & 0.16 & 1.17 & 0.22 & 2.84 & 0.10 & 1.28 & 0.09 & 1.24\\
		& 200 & 0.08 & 0.81 & 0.09 & 1.46 & 0.07 & 0.86 & 0.06 & 0.81\\
		& 500 & 0.04 & 0.50 & 0.03 & 0.71 & 0.03 & 0.55 & 0.04 & 0.50\\
		\hline
		\multirow{3}{*}{$d=3, K=30$} & 100 & 0.14 & 0.70 & 0.19 & 3.28 & 0.10 & 0.76 & 0.11 & 0.71\\
		& 200 & 0.07 & 0.48 & 0.13 & 1.76 & 0.05 & 0.52 & 0.06 & 0.48\\
		& 500 & 0.04 & 0.31 & 0.03 & 0.87 & 0.03 & 0.34 & 0.04 & 0.31\\
		\hline
		\multirow{3}{*}{$d=18, K=10$} & 100 & 1.88 & 11.06 & 1.90 & 35.18 & 1.61 & 10.67 & 1.39 & 11.08\\
		& 200 & 0.76 & 6.10 & 0.93 & 19.59 & 0.69 & 6.07 & 0.75 & 6.09\\
		& 500 & 0.28 & 3.50 & 0.31 & 11.34 & 0.25 & 3.54 & 0.28 & 3.50\\
		\hline
		\multirow{3}{*}{$d=18, K=30$} & 100 & 1.84 & 6.43 & 1.90 & 35.18 & 1.55 & 6.16 & 1.55 & 6.28\\
		& 200 & 0.74 & 3.53 & 0.93 & 19.59 & 0.65 & 3.51 & 0.74 & 3.53\\
		& 500 & 0.27 & 2.02 & 0.31 & 11.34 & 0.24 & 2.03 & 0.27 & 2.02\\
		\bottomrule
	\end{tabular}
\end{table}

\begin{table}
	\caption{\label{table: sim5}Bias and SE in Mendelian randomization with invalid instruments (results are multiplied by 10)}
	\centering
	\begin{tabular}{*{8}{c}}
		\toprule
		Estimator & MR-Egger & Weighted Median & IVW & Weighted Mode & RAPS& $\tilde{\theta}$ & $\hat{\theta}$ \\
		\midrule
		Bias & 6.48 & -2.05 & 14.44 & -1.30 & 16.95 & 0.71& -0.26 \\ 
		SE & 4.79 & 1.09 & 1.73 & 82.14 & 2.34 & 1.79 & 1.17\\
		\bottomrule
	\end{tabular}
\end{table}
~\\
	\newpage
	\begin{figure}[h]
		\subfigure[IVW estimation using all datasources for treatment effect in PD: red solid line; estimation for treatment effect in PD produced by $\tilde{\theta}$: green dashed line; estimation for treatment effect in PD produced by $\hat{\theta}$: blue dotted line.]{
			\includegraphics[scale = 0.45]{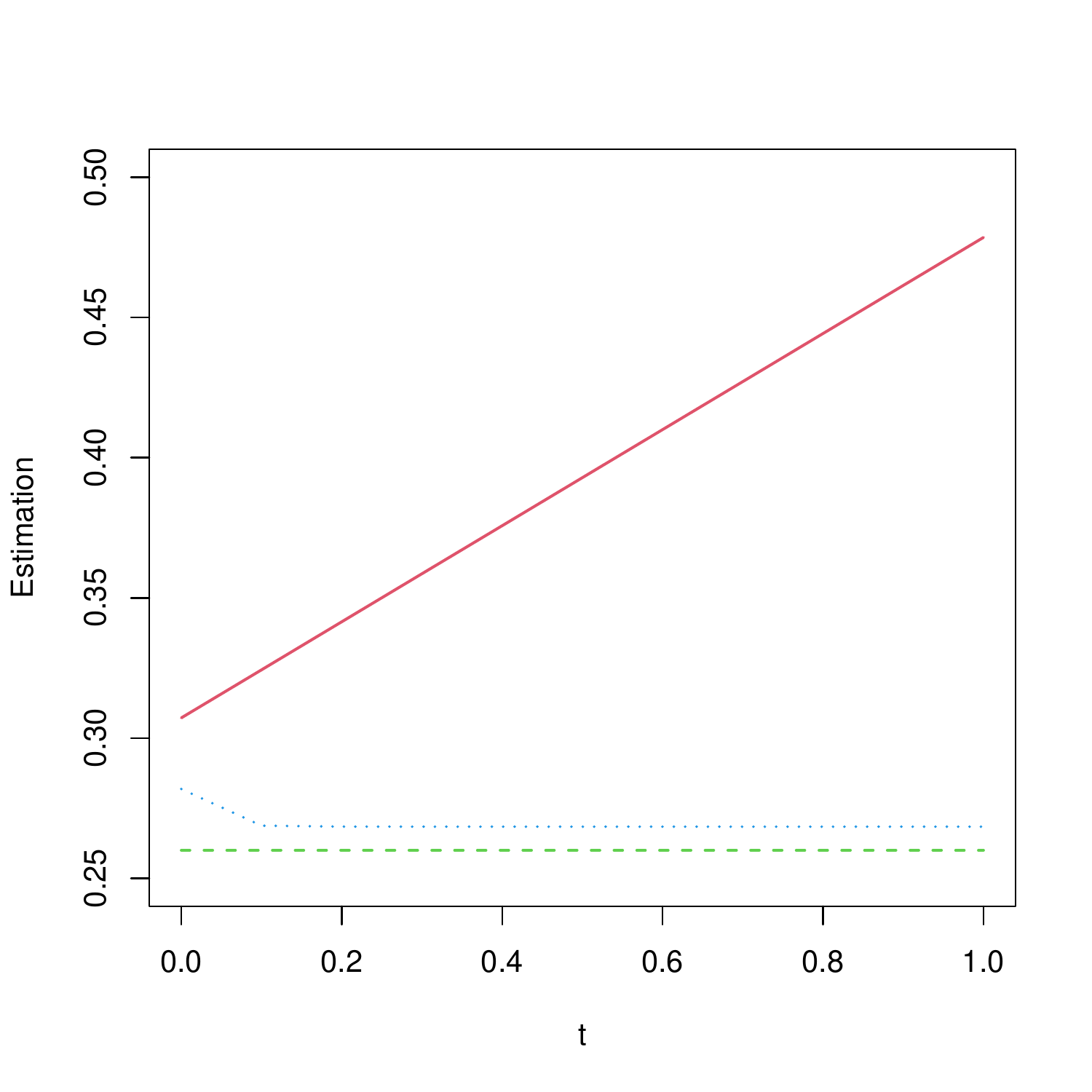}
		}\hspace{10mm}
		\subfigure[IVW estimation using all data sources for treatment effect in AL: red solid line; estimation for treatment effect in AL produced by $\tilde{\theta}$: green dashed line; estimation for treatment effect in AL produced by $\hat{\theta}$: blue dotted line.]{
			\includegraphics[scale = 0.45]{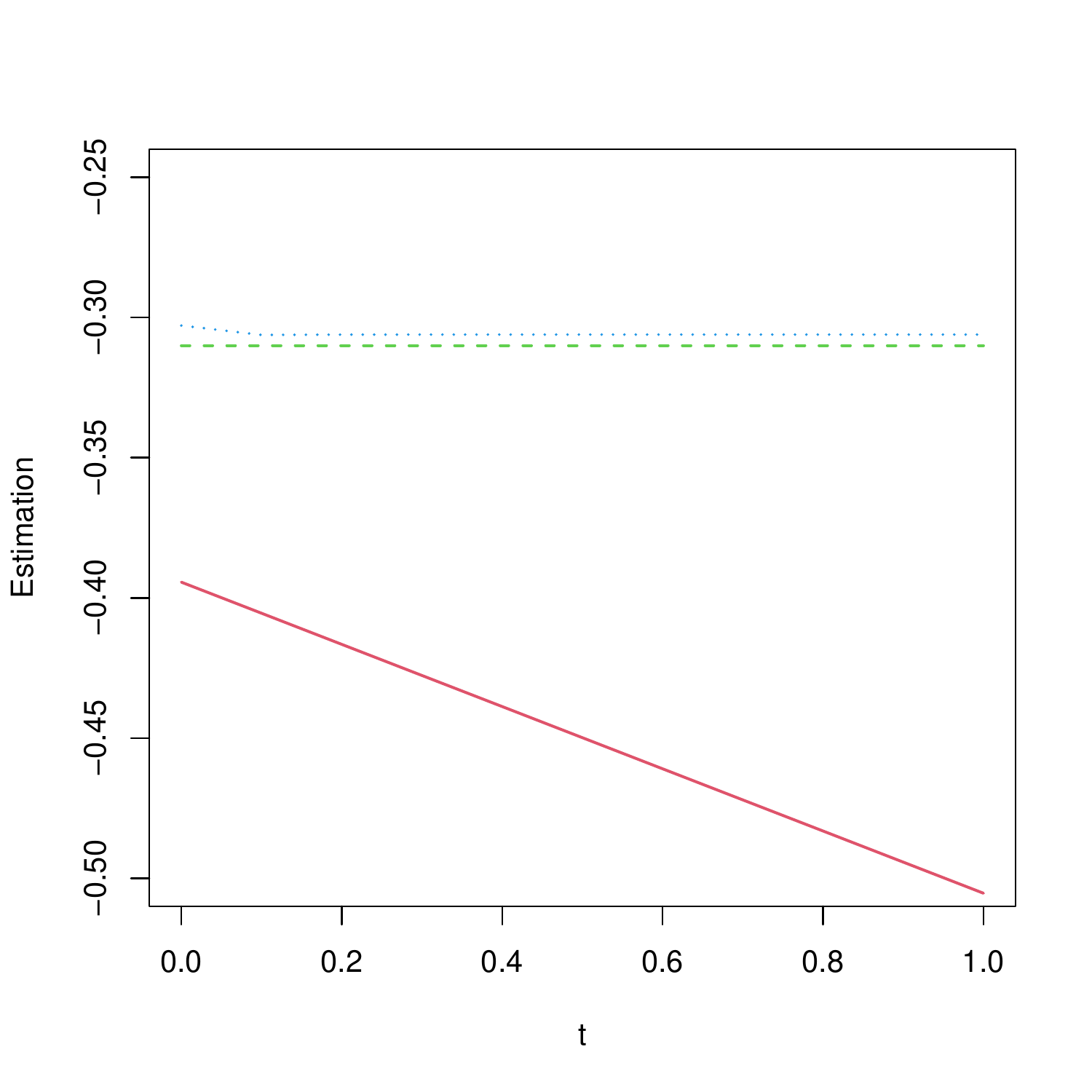}
		}
		\caption{Estimation results under different $t$.}\label{fig: stable}
	\end{figure}
%
\end{document}